\documentclass{article}

\usepackage{microtype}
\usepackage{subfigure}

\usepackage[accepted]{icml2021}
\usepackage{amsmath,amssymb,amsfonts}
\usepackage{isomath}
\usepackage{bm}
\usepackage{upgreek}
\usepackage{hyperref}
\usepackage{graphicx}
\usepackage{textcomp}
\usepackage{amsthm}
\usepackage{array}
\usepackage{soul}
\setstcolor{red}
\usepackage{booktabs}
\usepackage{multirow}
\usepackage[font=small,
    labelfont=bf,
    format=plain]{caption}
\usepackage{xcolor}
\newcommand{\x}{\mathbf{x}}
\newcommand{\y}{\mathbf{y}}
\newcommand{\z}{\mathbf{z}}
\newcommand{\f}{\mathbf{f}}
\newcommand{\w}{\mathbf{w}}
\renewcommand{\v}{\mathbf{v}}
\newcommand{\g}{\mathbf{g}}
\newcommand{\R}{\mathbb{R}}

\newcommand{\Wp}{{\mathcal{W}^+}}

\newcommand\numberthis{\addtocounter{equation}{1}\tag{\theequation}}
\newcommand{\vectheta}{\boldsymbol{\uptheta}}

\renewcommand{\u}{\mathbf{u}}

\renewcommand{\vec}[1]{\ensuremath{\mathbf{#1}}}

\providecommand{\norm}[1]{\ensuremath{\left\lVert#1\right\rVert}}

\def\[{\begin{equation}}
\def\]{\end{equation}}
\DeclareMathOperator*{\argmin}{arg\,min}

\newcommand{\fpi}{\ensuremath{\f^{(\rm{PI})}}}
\newcommand{\wpi}{\ensuremath{\w^{(\rm{PI})}}}

\theoremstyle{definition}
\newtheorem{definition}{Definition}[section]

\theoremstyle{theorem}
\newtheorem{theorem}{Theorem}[section]

\theoremstyle{lemma}
\newtheorem{lemma}{Lemma}[section]

\theoremstyle{notation}
\newtheorem{notation}{Notation}[section]

\theoremstyle{remark}

\def\equationautorefname~#1\null{Eq.~(#1)\null}

\icmltitlerunning{Prior Image-Constrained Reconstruction using Style-Based Generative Models}

\begin{document}

\twocolumn[
\icmltitle{Prior Image-Constrained Reconstruction using Style-Based Generative Models}



\icmlsetsymbol{equal}{*}

\begin{icmlauthorlist}
\icmlauthor{Varun A. Kelkar}{ece}
\icmlauthor{Mark A. Anastasio}{ece}
\end{icmlauthorlist}

\icmlaffiliation{ece}{University of Illinois at Urbana-Champaign, Urbana, IL 61801, USA}
% \icmlaffiliation{bioe}{Department of Bioengineering, University of Illinois at Urbana-Champaign, Champaign, IL 61801, USA}

\icmlcorrespondingauthor{Varun A. Kelkar}{vak2@illinois.edu}
\icmlcorrespondingauthor{Mark A. Anastasio}{maa@illinois.edu}

\icmlkeywords{Sparsity and compressed sensing, Generative models}

\vskip 0.3in
]



\printAffiliationsAndNotice{}  % leave blank if no need to mention equal contribution

\begin{abstract}
Obtaining a useful estimate of an object from highly incomplete imaging measurements remains a holy grail of imaging science. Deep learning methods have shown promise in learning object priors or constraints to improve the conditioning of an ill-posed imaging inverse problem. In this study, a framework for estimating an object of interest that is semantically related to a known prior image, is proposed. An optimization problem is formulated in the disentangled latent space of a style-based generative model, and semantically meaningful constraints are imposed using the disentangled latent representation of the prior image. Stable recovery from incomplete measurements with the help of a prior image is theoretically analyzed. Numerical experiments demonstrating the superior performance of our approach as compared to related methods are presented.
\end{abstract}

\section{Introduction}
In recent years, generative models based on deep neural networks have risen to the forefront of machine learning research. By taking advantage of the approximately low dimensional structure of natural image data, generative models, such as generative adversarial networks (GANs) have been able to learn effective mappings from a simple, tractable, low dimensional distribution to a complex image data distribution, giving it the ability to generate highly realistic looking images, approximately from the distribution induced by the training data. Apart from image synthesis, generative models have found applications such as density estimation \cite{realnvp}, image restoration \cite{srgan}, object detection \cite{object_detection}, video generation \cite{gan_video} and style transfer \cite{style_transfer} to name a few. Recent research in GANs has achieved state of the art performance in terms of visual quality of images generated, invertibility and image representation, and a meaningful control of semantic features of the image \cite{stylegan, stylegan2}.

Generative models have also found applications for image reconstruction in computed imaging systems, where a computational procedure is employed to form an estimate of an object of interest from imaging measurements. When the measurements are insufficient to uniquely determine the object of interest, this procedure amounts to solving an ill-posed inverse problem, and requires additional prior information about the object distribution. Approaches incorporating sparsity-based priors have been successful in obtaining accurate reconstructions from incomplete measurements \cite{candes_review, mri_lustig}. Seeking to better characterize the object distribution, generative models have been proposed as a prior for solving ill-posed inverse problems in imaging \cite{bora}. Improvements in the image-synthesis performance, stability, and quality of generative neural networks have in-turn improved the performance of generative-model constrained reconstruction methods.

In this study, style-based generative models
are investigated for constraining image reconstruction problems in which the solution is known to be close to a given prior image.
This scenario, known as \textit{prior image-constrained reconstruction}, is of particular significance when, for instance, the same object evolving over time is to be imaged multiple times. It finds applications in several scientific and medical imaging situations, such as monitoring tumor progression or perfusion \cite{piccs}, multi-contrast magnetic resonance imaging (MRI) \cite{refmri}, or sequential radar imaging \cite{piccs_radar}. Traditionally, this problem has been solved by assuming that the difference between the object of interest and the prior image is sparse in some domain \cite{piccs, cspi}, and solving an optimization problem penalizing this difference. However, although many natural images are compressible in some domain, their differences may not be so. Hence, penalizing the $\ell_1$ norm of the difference between the object and the prior image in a linear transform domain may not be the best strategy to ensure that the ground truth is related to the prior image in a meaningful way. Style-based generative models have been known to be able to control individual semantic features, or styles, in an image by varying the disentangled latent representation of the image at different scales. {In this work}, the inverse problem is formulated as an optimization problem in the disentangled latent space of a style-based generative model. The disentangled latent representation of the prior image is computed and prior image-based regularization is imposed by constraining the estimate to have certain styles equal to the corresponding styles of the prior image.

\textbf{Related work.} Generative model-constrained reconstruction has been an active area of research, after being first proposed by \cite{bora}. Several approaches have tried to reduce the representation error arising from the use of GANs as priors \cite{aliahmed, iagan}. Other works have examined denoising and image reconstruction when the measurement noise is distributed in a complex way \cite{inn_dependent_noise}. Recently, StyleGAN has been used for image superresolution \cite{pulse}, and image reconstruction in general \cite{brgm}. Individual control over image semantics has been achieved \cite{stylegan, stylerig} by controlling the disentangled latent respresentation of the image. Since the disentangled latent representation of an image is crucial, several studies have focused on StyleGAN inversion \cite{stylegan_inv_gaussian, image2stylegan}. {Regularization by way of imposing constraints based on the latent space structure has also been explored \cite{clinn}.}
Prior image-constrained compressed sensing has been studied previously from the theoretical \cite{cspi} point of view, to applications \cite{piccs, piccs_radar}. Some studies have used adaptive weights on the prior image and the image estimate to better model the differences between the ground truth and the prior image \cite{refmri}.

This study is organized as follows. \autoref{sec:bkd} describes the background of compressed sensing, generative model-constrained image reconstruction, prior image-constrained reconstruction and style-based generative models. \autoref{sec:approach} describes the proposed approach motivated by StyleGANs. \autoref{sec:theory} describes a theoretical analysis of the problem at hand. \autoref{sec:numerical} describes the setup of the proposed numerical studies, with the results being described in \autoref{sec:results}. Finally, discussion and conclusion is presented in Section \autoref{sec:concl}.
% for intro
%There is evidence to suggest that further improvements in generative model-constrained reconstruction may be possible with strategies such as using improved generative models \cite{aliahmed, brgm, pulse}, or 

% In recent years, improvements in the architecture of generative models has led to generative models that produce highly realistic images, or have low representation error, or better controllability over the type of images generated \cite{progan, glow, realnvp, biggan}.
% In particular, StyleGAN and StyleGAN2 are well known to possess all of these desirable characteristics \cite{stylegan, stylegan2}. 
% It has been demonstrated that these properties of modern generative models can be used to regularize inverse problems \cite{pulse, brgm, aliahmed, inn_dependent_noise, clinn}. 

\section{Background}\label{sec:bkd}
Several digital imaging systems can be approximately modeled by a linear imaging model, described as \cite{barrett}
\begin{align}
    \g = H\f + \vec{n},
\end{align}
where $\f \in \mathbb{E}^n$ is a vector that approximates the object to-be-imaged, $\g \in \mathbb{E}^m$ corresponds to the imaging measurements, and $\vec{n} \in \mathbb{E}^m$ represents the measurement noise. Here, $\mathbb{E}^l$ corresponds to an $l \in \mathbb{N}$ dimensional Euclidean space. $H \in \mathbb{E}^{m\times n}$ corresponds to the linear operator that approximates the underlying physical model of the imaging system. Often, the measurements are incomplete $(m < n)$ and, as such, $\f$ cannot be uniquely recovered from $\g$. In this case, in order to obtain a useful estimate of the true object, prior knowledge about $\f$ is needed to constrain the domain of $H$.

\subsection{Compressed sensing}
In recent decades, \textit{compressed sensing} has emerged as a leading framework to solve such underdetermined systems of equations. It achieves this by constraining $\f$ to a set of vectors that are sparse in some domain, and evoking certain conditions on $H$. Specifically, stable recovery of $k$-sparse objects can be guaranteed if $H$ satisfies the \textit{Restricted Isometry Property} (RIP) over the set of $2k$-sparse signals \cite{candes_romberg_tao, candes_review}. 

\begin{definition}[Restricted Isometry]
Let $S_{k}$ be the set of all $k$-sparse vectors in $\R^n$. $H$ is said to satisfy the RIP over $S_k$ if $\exists ~ \delta_k \in (0,1)$ that satisfies
\begin{align}
    (1-\delta_k) \norm{\f}_2^2 \leq \norm{H\f}_2^2 \leq (1+\delta_k)\norm{\f}_2^2,
\end{align}
for all $\f \in S_k$, and $\delta_k$ is not too close to 1 in a way prescribed in \cite{candes_romberg_tao}. 
\end{definition}

\subsection{Prior image-constrained reconstruction}
Prior image-constrained reconstruction is a scenario where the true object $\tilde{\f}$ is related, or close to, a previously known \textit{prior image} $\fpi$ \cite{piccs, cspi}. In traditional approaches to this problem, this similarity is imposed in the following way. In addition to the conventional sparsity constraint with respect to a transformation $\rm\Phi$, it is assumed that the difference of $\tilde{\f}$ and $\fpi$ is sparse with respect to a transform $\rm\Psi$, such as the wavelet transform or the 2D difference operator. Solving the inverse problem can then be cast as obtaining the solution to versions of the following optimization problem \cite{piccs, cspi, refmri}:
\begin{multline}\label{eqn:piccs}
    \hat{\f} = \argmin_{\f} \norm{\g - H\f}_2^2 + \lambda\left( \alpha\norm{\rm\Psi(\f - \fpi)}_1 \right.\\+ \left. (1-\alpha)\norm{\rm\Phi\f}_1 \right).  
\end{multline}

It is important to note that the assumption of sparse differences between the prior image and the ground truth may not always be valid. {This is because in real life, images are better modeled as being compressible in a transform domain and are structured \cite{cs_theory_new}. A difference of two such images may not only be denser than the original images, but also structured in a more complicated manner.} Hence, it is important to develop new ways to quantify the similarity between the sought-after and the prior images.

\subsection{Generative model-constrained reconstruction}
Generative model-constrained reconstruction is a framework for image reconstruction, where the domain of $H$ is constrained with the help of a generative model trained to approximate the distribution of objects \cite{bora, aliahmed, brgm}. Let $G : \R^k \rightarrow \R^n$ be a generative model, typically parametrized by a deep neural network with parameters $\vectheta$. $G$ is trained on a dataset of images, such that under $G$, a sample $\z \in \R^k$ from a tractable distribution such as $\mathcal{N}(0,I_k)$ maps to a sample $G(\z)$ that approximately comes from the distribution of the training dataset images. Here, $\z$ is called the latent representation of $G(\z)$. Since many real-life image datasets are approximately low dimensional, popular generative models, such as GANs often have a low dimensional domain with dimensionality $k \ll n$. Taking advantage of this fact, Bora \textit{et al.} proposed a way to guarantee stable reconstruction of an object in the range of a generative model $G$ having a Lipschitz constant $\bar{L}$ from $O(k\log(\bar{L}r/\delta))$ measurements, if the object's latent representation has an $\ell_2$ norm of at most $r$, and if $H$ satisfies the following \textit{set-restricted eigenvalue condition} (S-REC):
\begin{definition}[Set-restricted eigenvalue condition]\label{def:srec}
Let $S \subseteq \R^n$. A matrix $H \in \R^{m\times n}$ satisfies the set-restricted eigenvalue condition $\text{S-REC}(S, \gamma, \delta)$ for some constants $\gamma > 0$ and $\delta \geq 0$, if for any $\f_1, \f_2 \in S$,
\begin{align}
    \norm{H(\f_1 - \f_2)}_2 \geq \gamma\norm{\f_1 - \f_2}_2 - \delta.
\end{align}
\end{definition}
Intuitively, this property stipulates that two objects $\f_1$ and $\f_2$ in the range $\mathcal{R}(G)$ of $G$ may give rise to measurements under $H$ that are close, only if they themselves are close. Certain sensing matrices, such as iid Gaussian sensing matrices with an appropriate column length have been shown to satisfy the S-REC \cite{bora}. The guarantees of stable recovery are applicable to the solution of the following constrained optimization problem:
\begin{align*}\label{eqn:csgm_nonlagr}
    \hat{\z} &= \argmin_{{\z}, {\norm{\z}\leq r}} \norm{\g - HG(\z;\vectheta) }_2^2, \\
    \hat{\f} &\equiv G(\hat{\z}; \vectheta),\numberthis{}
\end{align*}
where $\g = H\tilde{\f} + \mathbf{n}$ is the measurement corresponding to the unknown true object $\tilde{\f}$. Since the above objective is non-convex, standard gradient descent-based methods are not guaranteed to converge to the optimal solution. However, it is observed that in practice, gradient-based methods give estimates of $\hat{z}$ that are close to the optimum, at least in the case when $\tilde{f} \in \mathcal{R}(G)$. Although Bora, \textit{et al.} show numerical studies using a deep convolutional GAN (DCGAN), an optimization problem similar to \autoref{eqn:csgm_nonlagr} can also be formulated for recent advanced GAN architectures, such as StyleGAN, with demonstrably improved empirical performance \cite{brgm}.

\subsection{Style-based Generative Adversarial Networks}
StyleGAN and its successor, StyleGAN2, are well known for producing highly realistic samples from a real-life natural image distribution. They are characterized by an architecture, that consists of two sub-networks - (1) a \textit{mapping network} $g_{\rm{mapping}} : \R^{k} \rightarrow \R^{k}$, and an $L$-layer \textit{synthesis network} $G : \R^{Lk} \rightarrow \R^n$. In the conventional image generation mode, the mapping network maps a sample $\z \in \mathcal{Z} \equiv \R^k$ from an iid standard normal distribution to a vector $\u \in \mathcal{W} \equiv \R^k$. The input $\w$ to the $L$-layer synthesis network $G$ is formed by stacking $L$ copies of $\u$ to form a $K = kL$ dimensional vector $\w \in \mathcal{W}^+ \equiv \R^K$. The $i$th copy of $\u$ represents the input to the $i$th layer of $G$, which controls the $i$th level of detail in the generated image. In addition to these, $G$ also takes as input a collection of latent-noise vectors $\eta$ that control minor stochastic variations of the generated image at different resolutions. The ability of a StyleGAN to control features of the generated image at different scales comes in part due to this architecture, and in part, due to the style-mixing regularization during training \cite{stylegan}. The latter loosely corresponds to evaluating the training loss using images generated by a ``mixed" $\w$ vector, formed by stacking the $\mathcal{W}$-space outputs of different realizations of $\z$. In addition to these basic characteristics, StyleGAN2 introduces path-length regularization, which aids in better conditioning of $G$ and reducing the representation error \cite{stylegan2}. For conventional image generation, a sampled $\w$ vector is degenerate, containing $L$ copies of $\u$, and hence lies in a $k$-dimensional subspace of $\mathcal{W}^+$. However, studies have shown that from the point-of-view of projecting an image to the range of $G$, utilizing the entire $\mathcal{W}^+$ space has benefits in terms of lower representation error \cite{stylegan_inv_gaussian}.

\section{Approach}\label{sec:approach}

\subsection{Prior image-constrained reconstruction using StyleGANs (PICGM)}
StyleGAN and StyleGAN2 are able to vary certain styles of an image while keeping certain other styles fixed. For example, for the StyleGAN trained on a dataset of faces, it is possible to vary the hairstyle and hair color while keeping the general structure of the face the same. For a medical imaging dataset such as a dataset of multi-contrast brain MRI images, it is possible to control the contrast, or exact placement of the folds, ventricles, and other fine scale features while keeping the general structure of the image the same. Hence, {comparing the sought-after and the prior images in the latent space of a StyleGAN is a natural approach to quantify the similarity between the two images.}

Motivated by such style-mixing properties of the StyleGAN, one possible way to formulate {the prior image-constrained inverse problem} is as follows. Let $G : \mathbb{R}^K \rightarrow \mathbb{R}^n$ denote the synthesis network of a trained StyleGAN2. Note that the domain of $G$ is taken to be the extended space $\Wp$. Let $\fpi = G(\wpi)$ denote a known prior image in the range of $G$. Then, the proposed measurement model can be written as
\begin{multline}\label{eqn:meas_model}
    \g = H\tilde{\f} + \vec{n}, \quad \tilde{\f} \in \{G(\w) ~ s.t. ~ \w_{1:p_1} = \wpi_{1:p_1}, \\\w_{p_2:K} = \wpi_{p_2:K}\},
\end{multline}
where $p_1, p_2$ are multiples of $k$, $1 \leq p_1 < p_2 \leq K$, $\tilde{\f}$ is the sought-after ground truth and $\w_{u:v}$ denotes the section of vector $\w$ from indices $u$ through $v$.

Even assuming that $\fpi$ and $\tilde{\f}$ are in-distribution images, there are some practical concerns about the measurement model described above. (1) In practice, $\fpi$ may not lie in $\mathcal{R}(G)$ since $\mathcal{R}(G)$ is a $K$-dimensional manifold in $\mathbb{R}^n$. (2) If
$\fpi \in \mathcal{R}(G)$, its disentangled latent representation $\wpi$ may lie in an unstable region of $\mathcal{W}^+$ {and the style mixing properties of $G$ may not apply to it.}
(3) If $\fpi \in \mathcal{R}(G)$ and the style mixing performance using $\wpi$ is consistent with that of images drawn from $G$,
$\tilde{\f}$, i.e. the sought-after object, may still not lie in the set $\{G(\w) ~ s.t. ~ \w_{1:p_1} = \wpi_{1:p_1}, \w_{p_2:K} = \wpi_{p_2:K}\}$ for any $p_1, p_2$. The following solutions are proposed in order to alleviate the aforementioned concerns. Concern (1) essentially refers to the representation error when approximating $\fpi$. It was observed that for in-distribution images, if a latent representation $\wpi$ in the extended $\Wp$ space is sought, even if the latent-noise vectors $\eta$ are not optimized over, a close approximation to $\fpi$ can be obtained by gradient-descent based optimization, except for minor stochastic detail represented by $\eta$. 

In an attempt to resolve concern (2), \cite{stylegan_inv_gaussian} observed that a transformed version of $\w$, given by $\v = \text{LReL}_{\alpha}(\w)$ approximately follows a multivariate gaussian distribution with a mean $\bar{\v} \in \R^K$ and covariance ${\rm\Sigma} \in \R^{K\times K}$. Here, $\text{LReL}_{\alpha}(.)$ denotes the leaky-ReLU nonlinear activation \cite{leakyrelu}, defined as
\begin{align}
    \text{LReL}_{\alpha}(\x)_i = \left\{\begin{matrix} x_i, \quad x_i \geq 0, \\
                                                \alpha x_i, \quad x_i < 0.
    \end{matrix}\right.
\end{align}
The value of $\alpha$ is the reciprocal of the scaling value for negative numbers included in the last leakyReLU layer in the mapping network $g_{\rm{mapping}}$. This means that it is possible to regularize the inversion of $G$ with the help of a Gaussian prior on $\v$ \cite{stylegan_inv_gaussian}. The inversion process can then be formulated in terms of the following optimization problem \cite{stylegan_inv_gaussian}
\begin{align*}
    \wpi &= \argmin_{\w} \norm{\fpi - G(\w)}_2^2 + \lambda\norm{\v-\bar{\v}}_{\rm\Sigma}^2,\\
    \text{s.t. } \v &= \text{LReL}_{\alpha}(\w),\numberthis{}\label{eqn:inversion}
\end{align*}
where $\norm{\x}_{\rm\Sigma}^2 = \x^{\top}{\rm\Sigma}^{-1}\x$ is used to impose a prior on $\v$ corresponding to a Gaussian distribution with mean $\bar{\v}$ and covariance $\rm\Sigma$. As observed in \cite{stylegan_inv_gaussian}, $\mathcal{R}(G)$-projected estimates of $\fpi$ obtained in this way inherit the style-mixing and stability properties of samples from $G$. The tradeoff between an accurate representation of $\fpi$ and the style-mixing properties is governed by the regularization parameter $\lambda$.

Concern (3) can be addressed by arguing that since $\tilde{\f}$ is an in-distribution image, it has minimal representation error when optimizing over $\Wp$. Also, $p_1, p_2$ can be treated as tunable regularization parameters, which manage the trade-off between imposition of the prior from $\fpi$, and consistency with the measurements $\g$.

\begin{algorithm}[tb]
   \caption{Projected Adam algorithm for minimizing the objective in \autoref{eqn:picgm}.}
   \label{alg:projadam}
\begin{algorithmic}
   \STATE {\bfseries Input:} Measurements $\g$, prior image latent $\wpi$, Regularization parameters $p_1, p_2, \lambda$, maximum iterations $n_{\rm{iter}}$.
   \STATE $\mathcal{L}(\w; \lambda)$ : Objective function from \autoref{eqn:picgm}.
   \STATE Initialize Adam optimizer parameters $(\alpha, \beta_1, \beta_2)$. (Default parameters were used).
   \STATE Initialize iteration number $t \leftarrow 0$.
   \STATE Initialize $\w^{[0]} \leftarrow \wpi$.
   \WHILE{$\w^{[t]}$ \textbf{not} converged}
   \STATE Adam update \cite{adam}:
   $$\w^{[t]} \leftarrow \texttt{ADAM}_{\alpha, \beta_1, \beta_2}(\mathcal{L}(\w^{[t]}; \lambda))$$
   \STATE Projection step:
   \begin{align*}
       \w^{[t]}_{1:p_1} &\leftarrow \w^{\rm(PI)}_{1:p_1}\\
       \w^{[t]}_{p_2:K} &\leftarrow \w^{\rm(PI)}_{p_2:K}
   \end{align*}
   \STATE $t \leftarrow t + 1$
   \ENDWHILE
\end{algorithmic}
\end{algorithm}

Taking into account the above arguments, the inverse problem associated with \autoref{eqn:meas_model} is formulated as the following optimization problem:
\begin{align*}
    \hat{\w} &= \argmin_{\w} \norm{\g - HG(\w)}_2^2 + \lambda\norm{\v-\bar{\v}}_{\rm\Sigma}^2,\\
    \text{s.t. } \w_{1:p_1} &= \wpi_{1:p_1}, \quad \w_{p_2:K} = \wpi_{p_2:K},\\
                 \v &= \text{LReL}_{\alpha}(\w).\numberthis{}\label{eqn:picgm}
\end{align*}
Although the above problem is non-convex, similar to previous works \cite{bora, aliahmed}, useful estimates $\hat{\f}$ can be obtained by iterative gradient descent-based optimization. \autoref{alg:projadam} shows the projected-Adam algorithm used for this purpose \cite{adam}.

\subsection{Compressed sensing using StyleGAN2 (CSGM)}
Analogous to \autoref{eqn:picgm}, CSGM using StyleGAN2 can be formulated in the following manner:
\begin{align*}
    \hat{\w} &= \argmin_{\w} \norm{\g - HG(\w)}_2^2 + \lambda\norm{\v-\bar{\v}}_{\rm\Sigma}^2,\\
    \text{s.t. } \v &= \text{LReL}_{\alpha}(\w).\numberthis{}\label{eqn:csgm_stylegan}
\end{align*}
The Adam algorithm was used for obtaining approximate solutions to the above problem \cite{adam}.

\section{Theoretical Analysis}\label{sec:theory}
The theoretical analysis presented here is motivated by the results presented in \cite{bora}. There, the authors provide a stable recovery guarantee in terms of the Lipschitz constant of the generative network. However, in practice, Lipschitz constants can be difficult to estimate, or a generative network may not be Lipschitz stable. Hence, theoretical analysis in terms of the Jacobians of the generative network is presented here, in order to derive a limited, non-uniform guarantee for the stable recovery of typical in-distribution objects that lie in the range of the generative network. In order to do so, we utilize properties of StyleGAN2. 

Let $p_{\w}$ denote the distribution of the extended latent space vector $\w = [\u_1^\top ~ \u_2^\top ~\dots ~ \u_L^\top]^\top \in \Wp$, where $\u_i \sim g_{\rm mapping}(\z_i)$, $\z_i \sim \mathcal{N}(0,I_k)$, $\z_i$'s are independently distributed. The following assumptions are made based on the StyleGAN2 properties:

\textit{1) Path length regularity:}
    \begin{align*}
    \mathbb{E}_{\w \sim p_\w} \big( \norm{J(\w)}_F - a \big)^2 < b,\tag{AS1}\label{eqn:app_as1}
    \end{align*}
    $J(\w)$ denotes the Jacobian of $G$ evaluated at $\w$, $b > 0$ and $a = \mathbb{E}_\w \norm{J(\w)}_F$ are global constants. This assumption is inspired by the path-length regularization used in \cite{stylegan2}.
    
\textit{2) Approximate local linearity:}
    \begin{align*}
        \mathbb{E}_{\w \sim p_\w} &\max_{\begin{smallmatrix}\w'\\ \norm{\w'-\w}\leq \epsilon\end{smallmatrix}} \mathcal{L}(\w', \w) \leq \beta^2(\epsilon),\tag{AS2}\label{eqn:local_lin1}\\
        \intertext{where}
        \mathcal{L}(\w',\w) &= \left\|G(\w') - G(\w) - J(\w)(\w' - \w) \right\|_2^2,
    \end{align*}
    and $\beta$ is a positive function with $\beta(0) = 0$. This property essentially measures how close to linear $G$ behaves in an $\epsilon$-neighborhood around a point $\w$.

As described in \cite{stylegan2}, estimates of $a,b$ can be computed via an empirical estimation of $\mathbb{E}_{\w}\mathbb{E}_{\y \sim \mathcal{N}(0,I_n)} \norm{J(\w)^\top \y}_2$. Estimating $\beta(\epsilon)$ for a given value of $\epsilon$ is however not tractable. Nevertheless, an approximate estimation can be obtained by first computing the Jacobian at a point $\w \sim p_\w$, and then iteratively maximizing $\mathcal{L}(\w',\w)$ using a projected gradient ascent-type algorithm followed by an empirical mean of the maxima for several such $\w \sim p_\w$. Empirically estimated values of $a,b$ and $\beta$ are presented in the supplementary file.

Using these two assumptions, the following results are derived.

\begin{notation} 
Let $p_\v$ denote the distribution of $\v = \text{LReL}_\alpha(\w)$, and let $\bar{\v}$, $\rm\Sigma$ be its mean and covariance matrix respectively.
Let $\wpi$ be a sample from $p_\w,$ and $1 \leq p_1 < p_2 \leq K.$ Assume that $p_1$ and $p_2$ are multiples of $k$. Let
    \begin{multline*}
        B_\w^{p_1,p_2} (r) := \Big{\{}\w ~ s.t. ~ \norm{\mathrm{LReL}_\alpha(\w) - \bar{\v}}_{\rm\Sigma} \leq r, \\
        \quad\quad\quad\quad\quad\quad\quad\quad \w_{1:p_1}=\wpi_{1:p_1}, \w_{p_2:K} = \wpi_{p_2:K}\Big{\}},
    \end{multline*}
\end{notation}

By Markov's inequality and concentration of norm, we have the following \cite{hdp_book}:
\begin{lemma}\label{lem:ind_properties}
If $\w$ is a sample from $p_\w$, then it satisfies the following three properties with probability at least $1 - O(1/K)$:
\begin{align*}
\norm{J(\w)}_F &\leq \sqrt{K}a, \tag{P1}\\
\max_{\begin{smallmatrix}\w'\\ \norm{\w'-\w}\leq \epsilon\end{smallmatrix}} \mathcal{L}(\w', \w) &\leq \sqrt{K}\beta(\epsilon)\tag{P2}\\
\norm{\Sigma^{-1/2}\mathrm{LReL}_{\alpha}(\w)}_2 &\leq \sqrt{K}(1+o(1))\tag{P3}
\end{align*}
\end{lemma}

\begin{lemma}[Set-restricted eigenvalue condition]\label{lem:srec}
Let $\wpi$ be a sample from $p_\w$. Assume that $p_\v$ is a Gaussian distribution with covariance matrix $\rm\Sigma$.
Let $\tilde{B}^{p_1,p_2}_\w(r)$ be the set of all points in $B^{p_1,p_2}_\w(r)$ satisfying properties P1 and P2.
Let $\tau < 1$, $\delta > 0$. Let $H \in \mathbb{R}^{m\times n}$ be an matrix with elements $h_{ij} \sim \mathcal{N}(0, 1/m)$. For all $\delta'<\delta$, let $\beta(\delta'/a)$ go polynomially as $\delta'/a$, with $\beta(0) = 0$. If
\begin{align}
    m = {\rm\Omega}\left(\frac{p_2-p_1}{\tau^2}\log \frac{ar\norm{\rm\Sigma}_F}{\delta}\right),
\end{align}
then $H$ satisfies the $\text{S-REC}(G(\tilde{B}^{p_1,p_2}_\w(r)), 1-\tau, \delta+\sqrt{K}\beta(\delta/a))$ with probability $1 - e^{-{\rm\Omega}(\alpha^2m)}$.
\end{lemma}

Based on this, the following reconstruction guarantee is arrived at.

\begin{theorem}\label{thm:picgm}
Let $H \in \mathbb{R}^{m\times n}$ satisfy $\textrm{S-REC}(G(\tilde{B}^{p_1,p_2}_\w(r)), \gamma, \delta+\sqrt{K}\beta(\delta/a))$. Let $\vec{n}$ be the measurement noise. Let $\w, \wpi \sim p_\w$. Let $\fpi = G(\wpi)$ be the known prior image. Let 
$$\tilde{\w} = [\wpi_{1:p_1}{}^\top \quad \w_{p_1:p_2}^\top \quad \wpi_{p_2:K}{}^\top]^\top.$$
Let $\tilde{\f} = G(\tilde{\w})$ represent the object to-be-imaged. Let $\g = H\tilde{\f} + \vec{n}$ be the imaging measurements. Let
\begin{align}\label{eqn:unrelaxed}
    \hat{\f} = \argmin_{\f\in G(\tilde{B}^{p_1,p_2}_\w(r))} \norm{\g - H\f}_2^2.
\end{align}
Then, 
\begin{align}
    \| \hat{\f} - \tilde{\f} \| \leq \frac{1}{\gamma} (2\norm{\vec{n}} + \delta + \sqrt{K}\beta(\delta/a))
\end{align}
with probability $1 - O(1/K)$.
\end{theorem}

Note that instead of optimizing over $\f\in G(\tilde{B}^{p_1,p_2}_\w(r))$, \autoref{eqn:picgm} proposes to solve a relaxed version, over $G(B^{p_1,p_2}_\w(r))$. In addition to this, as mentioned previously, \autoref{eqn:picgm} is non-convex and convergence is not guaranteed. Due to this, the proposed approach may not always yield an optimum to \autoref{eqn:unrelaxed}. However, whether or not the solution lies in $G(\tilde{B}^{p_1,p_2}_\w(r))$ can always be checked by checking if P1, P2, P3 are satisfied. For estimating in-distribution objects in the range of $G$, it was observed empirically that the conditions P1, P2, P3 are satisfied by the estimated $\hat{\w}$.

\begin{figure}[ht]
\begin{center}
\centerline{\includegraphics[width=\columnwidth]{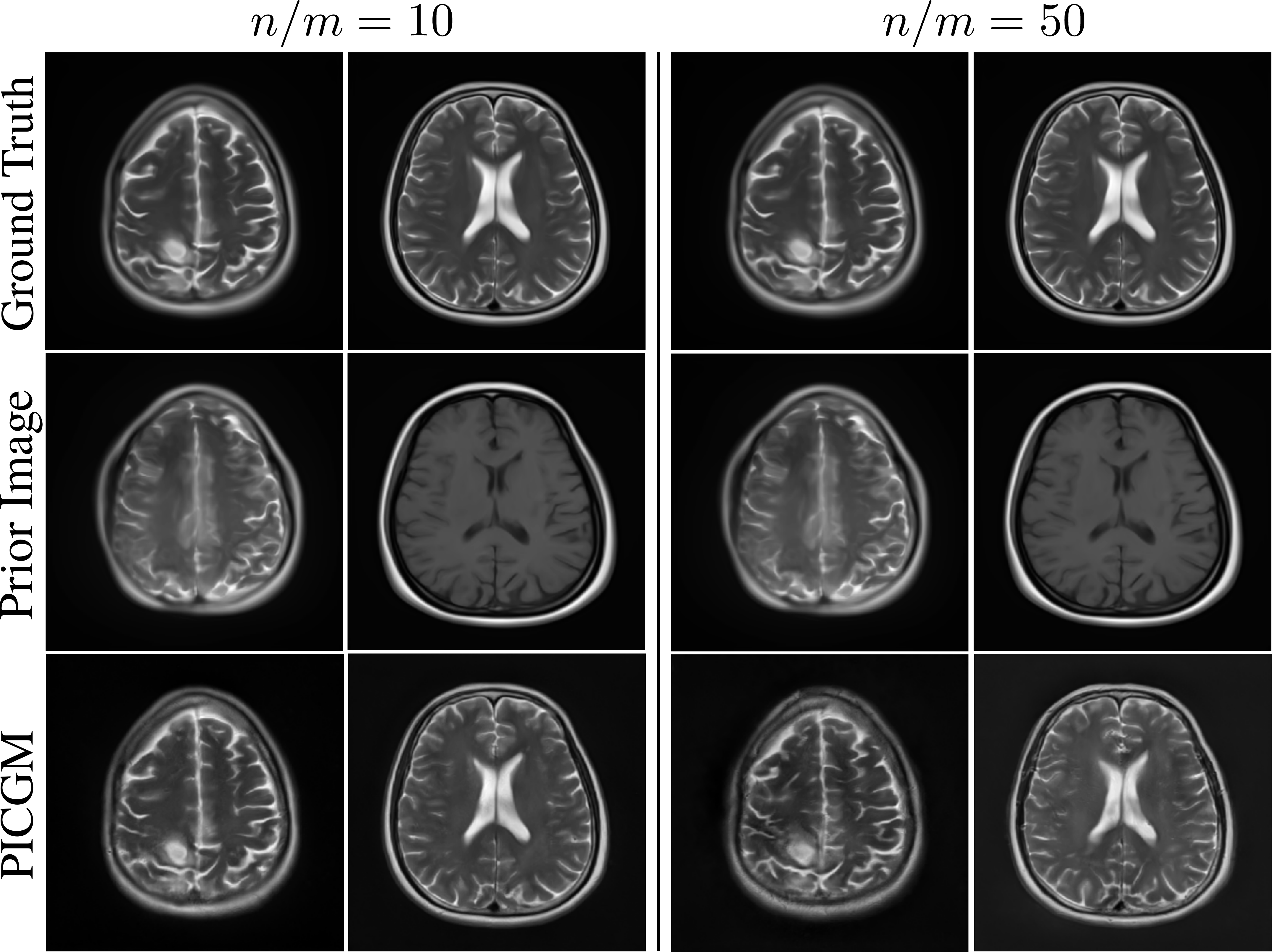}}
\caption{Ground truth, prior image and image estimated from Gaussian measurements with $n/m = 10$ and $n/m = 50$ using the proposed approach in the inverse crime case.}
\label{fig:inversecrime_images_gaussian}
\end{center}
\vskip -0.2in
\end{figure}

\begin{figure}[ht]
\begin{center}
\centerline{\includegraphics[width=\linewidth]{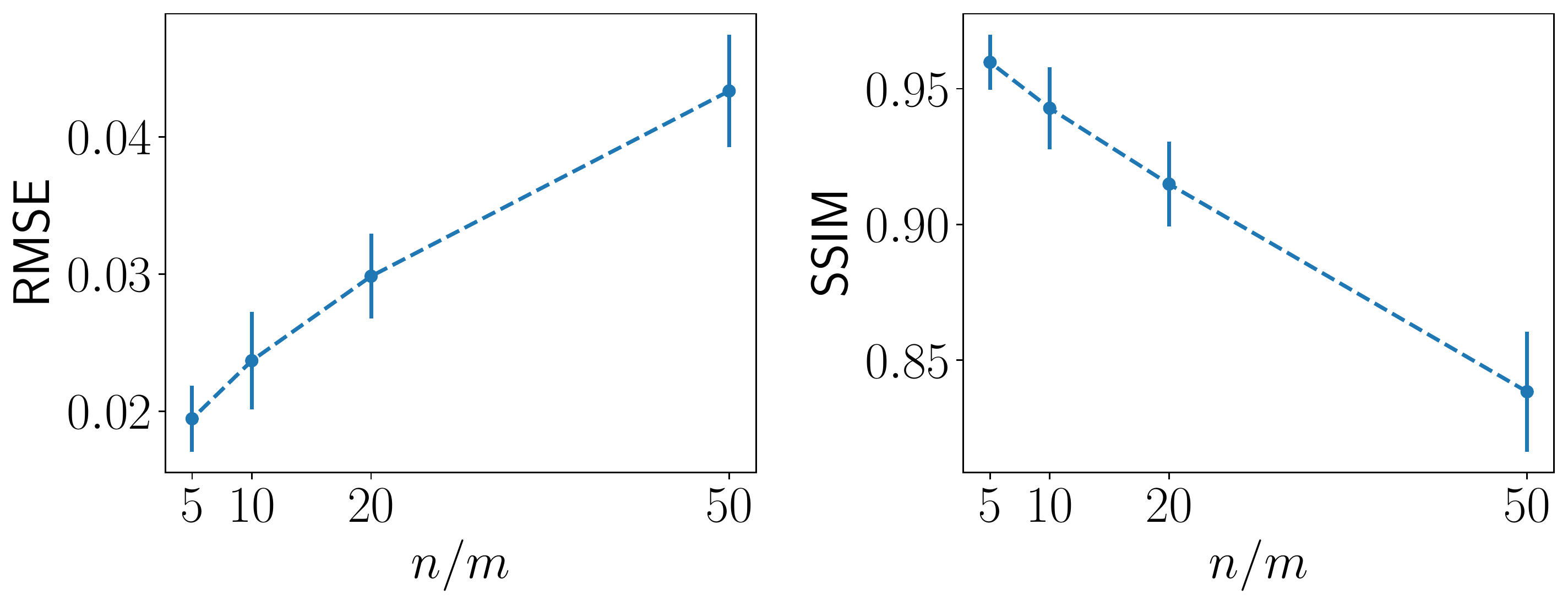}}
\caption{Ensemble RMSE and SSIM values for in the inverse-crime setting, for various subsampling ratios. The error bars show a span of one standard deviation.}
\label{fig:inversecrime_rmse_ssims}
\end{center}
\vspace{-10pt}
\end{figure}

\section{Numerical Studies}\label{sec:numerical}
{The numerical studies were split into three parts - (1) inverse-crime study, where the object was directly sampled from the StyleGAN2 and measurements were simulated using a Gaussian forward model, (2) face image study, where real face images were used to simulate noisy measurements using a Gaussian forward model, and (3) MR image study, where real brain MR images were used to simulate stylized undersampled MRI measurements with noise. Additionally, the robustness of the proposed approach to misalignment between the ground truth and the prior image was numerically assessed in the MR image study.  }

\textit{1) Dataset and forward model:}

{For the inverse-crime study, StyleGAN2 was trained on a composite brain MR image dataset consisting of a total of 200676 T1 and T2 weighted images of size 256$\times$256 from the fastMRI initiative database \cite{fastmri} and 866 images from the brain tumor progression dataset \cite{tumorprog}. These images were resized to 256$\times$256. For the face image study, a StyleGAN2 with an output image size of 128$\times$128$\times$3 was trained on images from the Flickr-Faces-HQ (FFHQ) dataset \cite{stylegan}. For the MR image study, StyleGAN2 was trained on a composite brain MR image dataset consisting of 164741 T1 and T2 weighted images from the fastMRI database, 686 images from the Brain tumor progression dataset, 2206 T1 and T2 weighted images from the TCIA-GBM dataset \cite{tcia_gbm}, and 36978 T2 weighted images from the OASIS-3 dataset \cite{oasis}.}

For evaluating reconstruction performance in the inverse crime studies, iid images were sampled from the generative model, where each style vector $\u_i$ was sampled independently from the mapping network, and the composite vector $\w = [\u_1, \u_2, \dots, \u_L]$ was used to sample the image. This was treated as the prior image. Four of the style vectors, $\u_6 \dots \u_9$ were then replaced by new styles $\u'_6 \dots \u'_9$, which were each again sampled independently using the mapping network. This was treated as the ground truth object to-be-recovered. Reconstruction performance was evaluated on a dataset of 50 such images. An iid Gaussian matrix $H \in \mathbb{R}^{m\times n}$ was used as the forward model and real-valued iid Gaussian noise with signal-to-noise ratio (SNR) 20 dB was added to the measurements. The reconstruction performance was evaluated for undersampling ratios $n/m = 5, 10, 20$ and $50$.

{For evaluating reconstruction performance in the face image study, 17 pairs of face images corresponding to differences in hairstyles, hair color and facial expressions were obtained from the stock-image hosting service Shutterstock \cite{shutterstock}. Images were manually inspected to ensure that one of the images in the pair was not a digitally altered version of the other. The images were manually cropped and resized to 128$\times$128$\times$3. In the numerical experiments, one of the images in the pair was used as the prior image. An iid Gaussian matrix with $n/m = 50$ was used to generate the measurements from the second image, which was treated as the ground truth. Real-valued iid Gaussian noise with 20 dB SNR was added to the measurements.}

{The brain tumor progression dataset contains pairs of brain MR images of test subjects, separated by time on the scale of a few months. The second image in the pair shows progression of the brain tumor with respect to the first image.} For evaluating the reconstruction performance in the case of the MR image study, a dataset of 22 held-out image pairs from the tumor progression dataset were used. The first image was used as the prior image. Imaging measurements were simulated from the second image, which was treated as the ground truth. A Fourier undersampling forward model, described as
$
H = \mathbf{m} \odot \mathcal{F} \in \mathbb{R}^{m\times n}
$
was used to simulate MRI measurements. Here, $\mathbf{m}$ corresponds to a binary mask, and $\mathcal{F}$ corresponds to the 2D discrete Fourier transform. This forward operator fully samples a fraction of the lower frequencies and randomly subsamples a fraction of high frequencies of the image. Five different undersampling ratios are considered: $n/m = 2, 4, 6, 8$ and $12$. Complex iid Gaussian noise with 20 dB SNR was added to the measurements.

\textit{2) Generative network training details:}
The StyleGAN2 architecture proposed in \cite{stylegan2} was used. For an image size of {$2^i\times2^i$, it contains $2(i-1)$ layers split across $i-1$ resolution levels.} The default latent space dimensionality of 512 was maintained. The networks were trained using Tensorflow 1.14/Python \cite{tensorflow} on an Intel Xeon E5-2620v4 CPU @ 2.1 GHz and four Nvidia TITAN X graphics processing units (GPUs)\footnote{Weights of the trained networks for our models can be found at: \url{https://databank.illinois.edu/datasets/IDB-4499850}}.

\textit{3) Baselines:}
For obtaining estimates of the true object from $\g$, the
performance of the following reconstruction methods were qualitatively and quantitatively compared -- (1) Penalized least squares with TV regularization (PLS-TV), (2) compressed sensing using StyleGAN2 (CSGM) mentioned in \autoref{eqn:csgm_stylegan}, (3) prior image-constrained compressed sensing (PICCS) mentioned in \autoref{eqn:piccs}, and (4) the proposed method (PICGM) introduced in \autoref{eqn:picgm}. Note that the first two  methods described do not utilize information from the prior image, while the last two do. For PICCS, discrete difference operator (corresponding to TV semi-norm) was used as the sparsifying transform $\rm \Phi$, and for the transform $\rm\Psi$, a 2D Haar wavelet transform of level 7 (i.e. equal to the number of resolution levels in StyleGAN2) was utilized. For evaluating PICGM, the latent representation $\wpi$ of the prior image was computed with the help of the trained StyleGAN using the procedure introduced in \autoref{eqn:inversion}. {The regularization parameters for all the methods were tuned using either a line search or a grid search depending upon the number of regularization parameters, and the setting giving the lowest ensemble root mean-squared error (RMSE) was chosen. The image estimates obtained by use of each of the four methods were quantitatively evaluated using the RMSE and structural similarity (SSIM). All algorithms were implemented using Tensorflow 1.14/Python.\footnote{The Tensorflow/python implementation of the reconstruction methods can be found at \url{https://github.com/comp-imaging-sci/pic-recon}}}

\textit{4) Testing the robustness to misalignments:}
The robustness of PICGM as compared to PICCS was numerically analyzed for the MR image study, by considering configurations where the ground truth is misaligned with the prior image. Five such configurations were considered. For each configuration $C_i$, the test ground truth images were each first rotated by $2i$ degrees clockwise or anticlockwise randomly with equal probability, and then translated by $2i$ pixels in a uniformly random direction. Estimates of the new ground truths were obtained with the help of PICCS and PICGM from 8 fold undersampled Fourier measurements, and the ensemble RMSE and SSIM performances of the two algorithms were compared.

\begin{figure}[ht!]
\vspace{-5pt}
\begin{center}
\centerline{\includegraphics[width=\columnwidth]{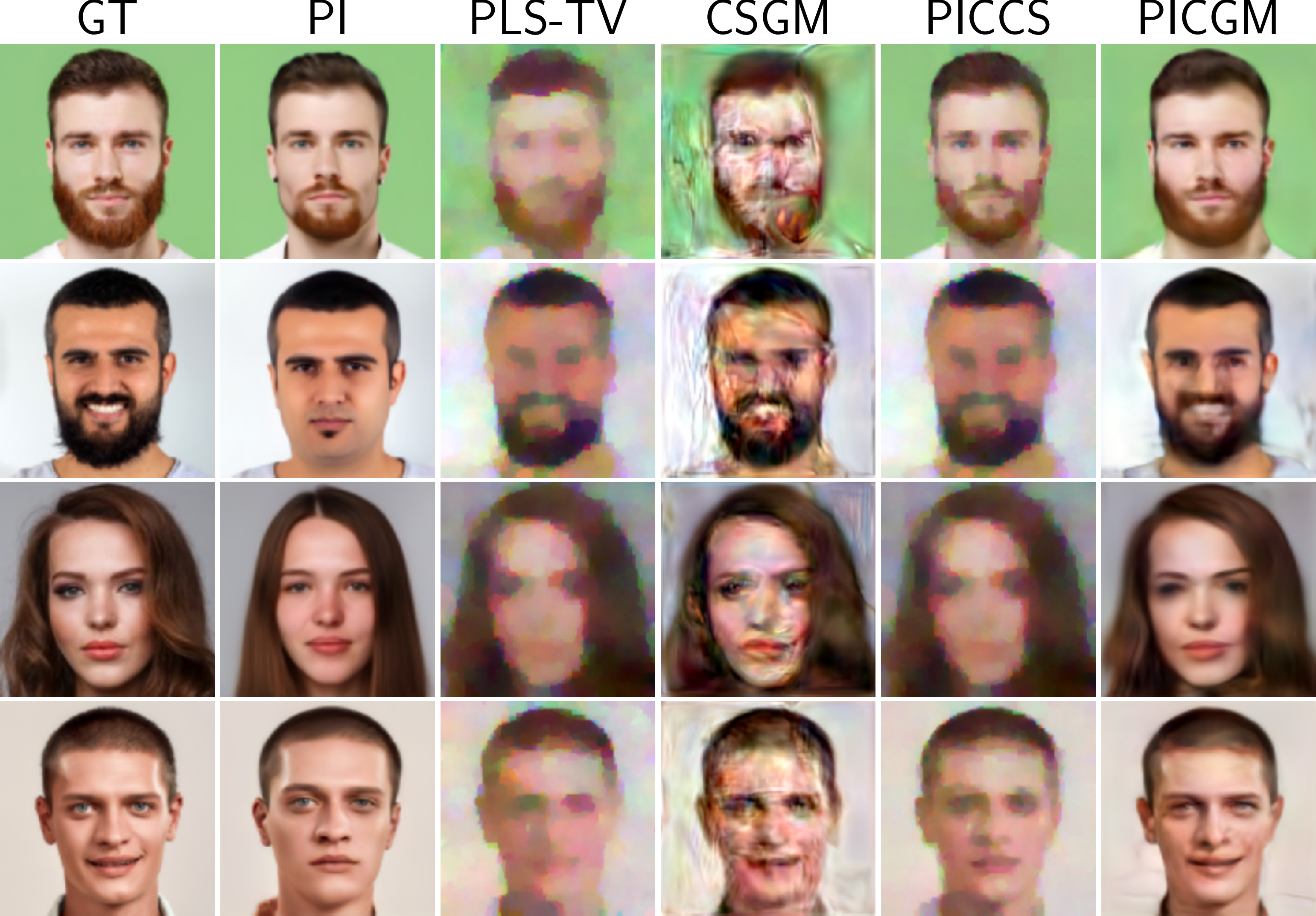}}
\caption{Ground truths (GT), prior images (PI) and image estimated from Gaussian measurements with $n/m = 50$ using the proposed approach in the face image study (zoom in for clarity).}
\label{fig:face_images_gaussian_0.02}
\end{center}
\vskip -0.2in
\end{figure}

\begin{figure}[ht!]
\vspace{-10pt}
\begin{center}
\centerline{\includegraphics[width=\columnwidth]{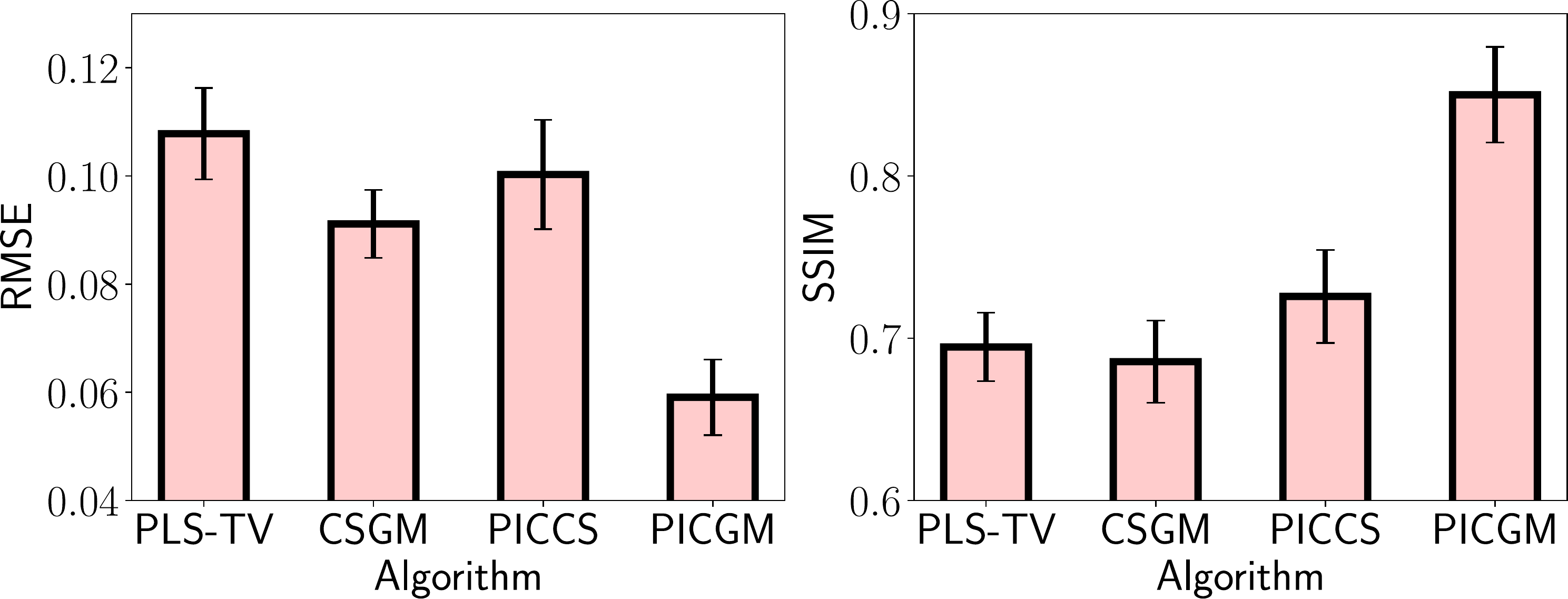}}
\caption{Ensemble RMSE and SSIM values for in the face image study for $n/m = 50$. The error bars span one standard deviation.}
\label{fig:faces_rmse_ssim}
\end{center}
\vskip -0.2in
\end{figure}

\begin{figure*}[ht]
\vskip 0.2in
\begin{center}
\centerline{\includegraphics[width=\linewidth]{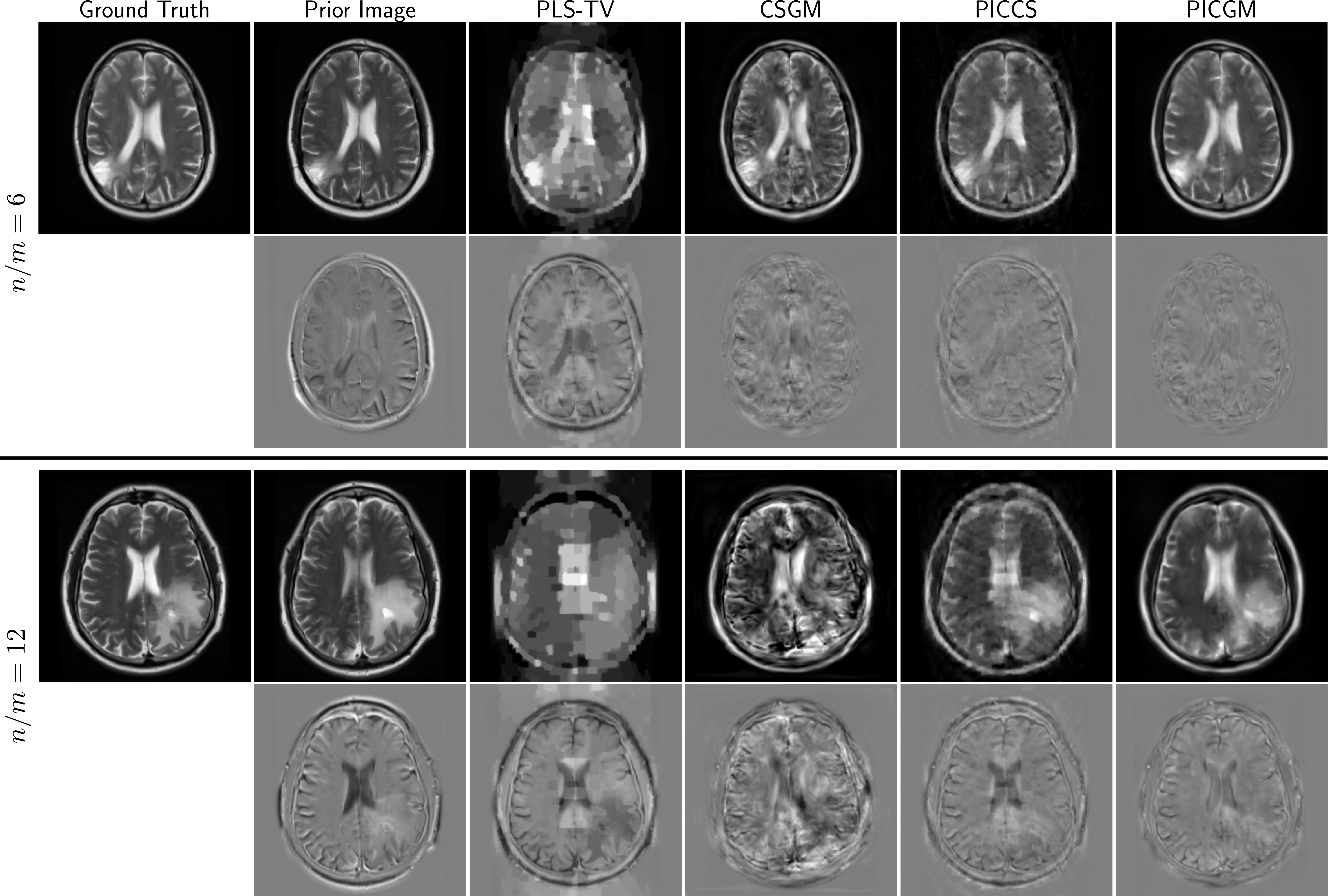}}
\caption{Ground truth, prior image, and images reconstructed from simulated MRI measurements with $n/m = 6$ and $n/m = 12$ along with difference images for the MR image study}
\label{fig:images_mask_rand_6x_12x}
\end{center}
\vskip -0.2in
\end{figure*}

\begin{figure}[ht]
\vskip 0.2in
\begin{center}
\centerline{\includegraphics[width=\linewidth]{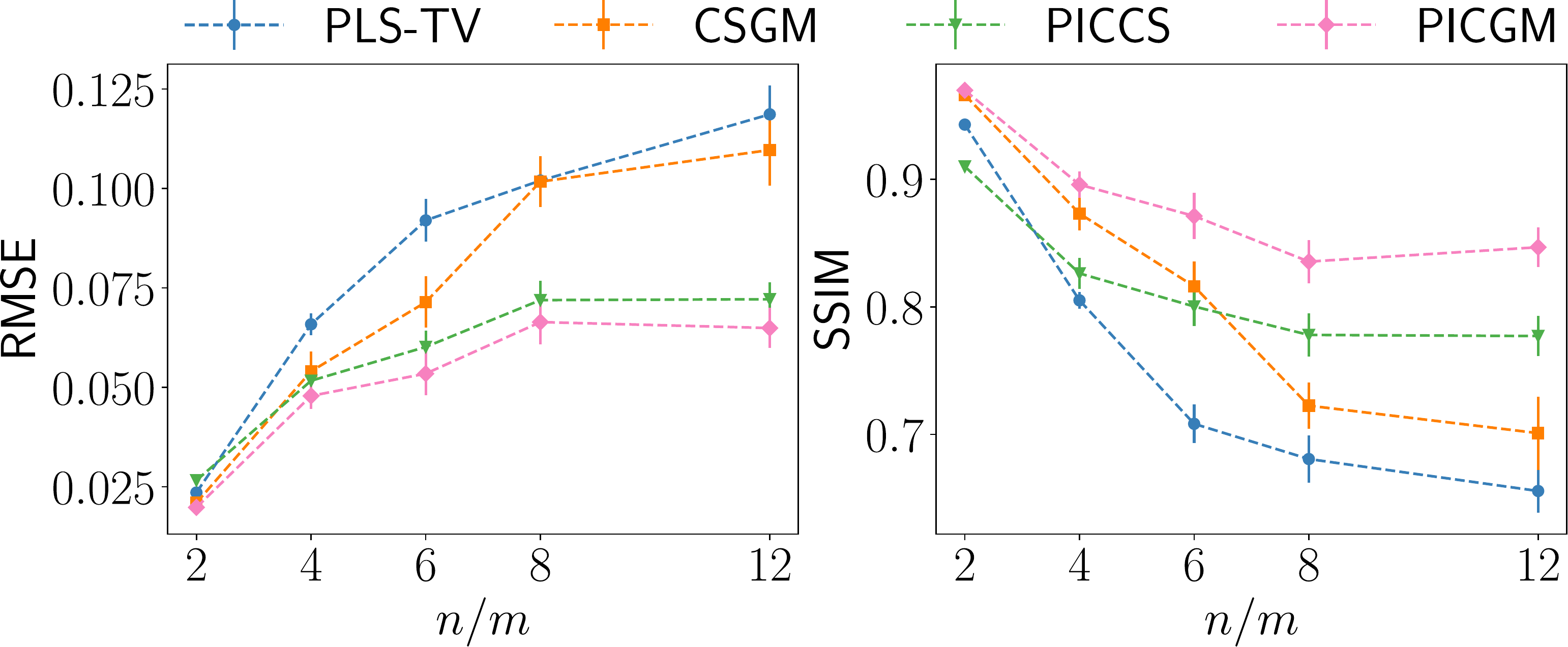}}
\caption{Ensemble RMSE and SSIM values for in the MR image study. The error bars show a span of one standard deviation.}
\label{fig:tumorprog_rmse_ssim}
\end{center}
\vskip -0.2in
\end{figure}

\section{Results}\label{sec:results}

Figure \ref{fig:inversecrime_images_gaussian} shows some of the reconstructed images for the inverse crime case, with $n/m = 10$ and $n/m = 50$, respectively. The RMSE and SSIM values over the ensemble of images are shown in \autoref{fig:inversecrime_rmse_ssims}. It can be seen that the proposed method performs well in terms of RMSE and SSIM even in the severe undersampling case, such as $n/m = 50$. 

{Figure \ref{fig:face_images_gaussian_0.02} shows some of the reconstructed images from the face image study. The ensemble RMSE and SSIM values are shown in \autoref{fig:faces_rmse_ssim}. It can be seen that the proposed algorithm outperforms the other algorithms considered in terms of perceptual quality, as well as RMSE and SSIM. This is because it is able to better capture the semantic differences between the sought-after and the prior image. }

Figure \ref{fig:images_mask_rand_6x_12x} shows the images reconstructed from simulated MRI measurements with undersampling ratios $n/m = 6$ and $n/m = 12$. Although PLS-TV and CSGM have difficulty in properly recovering the image, CSGM has noticeably better defined boundaries. Also, note that the ground truth and the prior image are visually similar in certain regions, and different in others. It can be seen that visually, the performance of the proposed method is the best. Ensemble RMSE and SSIM values for the MR image study are shown in \autoref{fig:tumorprog_rmse_ssim}. The results of the robustness study as shown in \autoref{fig:misalignment} indicate that while the proposed approach is robust to misalignments of the ground truth with respect to the prior image in the form of translations and rotations, the PICCS

\begin{minipage}{0.48\linewidth}
method breaks down. An example for $i = 4$ is shown in \autoref{fig:robustness_fig}. Additional images for the various studies are included in the supplementary file. \autoref{alg:projadam} takes around 5 minutes to converge to an optimal solution on a single Nvidia 1080 TX GPU. \\
\end{minipage}
\noindent\hfill\begin{minipage}{0.48\linewidth}
\captionsetup{type=figure}
\includegraphics[width=\linewidth]{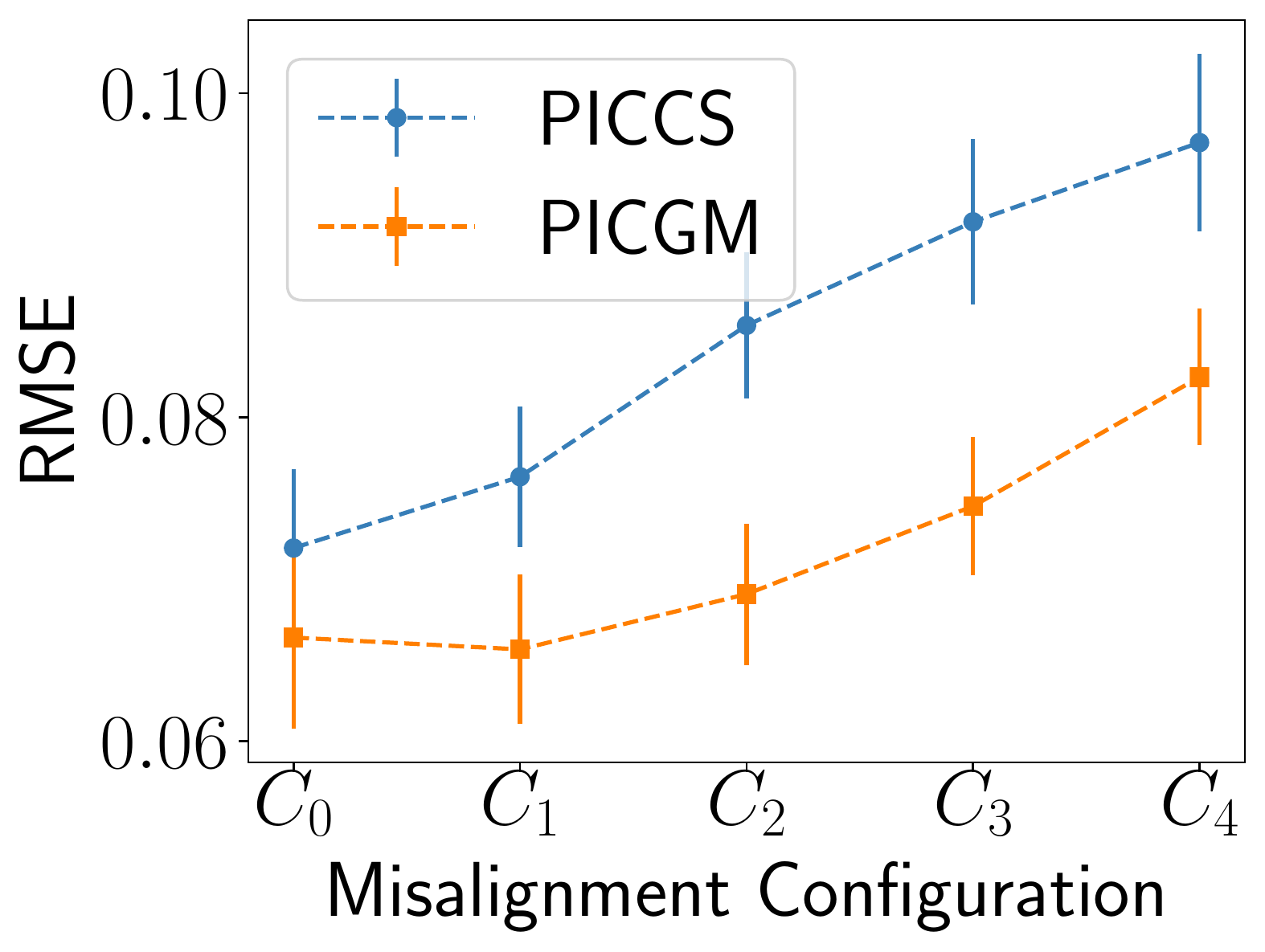}
%\vspace{0.01cm}
\captionof{figure}{Ensemble RMSE values for the MRI image study, with misaligned ground truth.}
\label{fig:misalignment}
\end{minipage}

\begin{figure}[ht]
% \vskip 0.2in
\begin{center}
\centerline{\includegraphics[width=\linewidth]{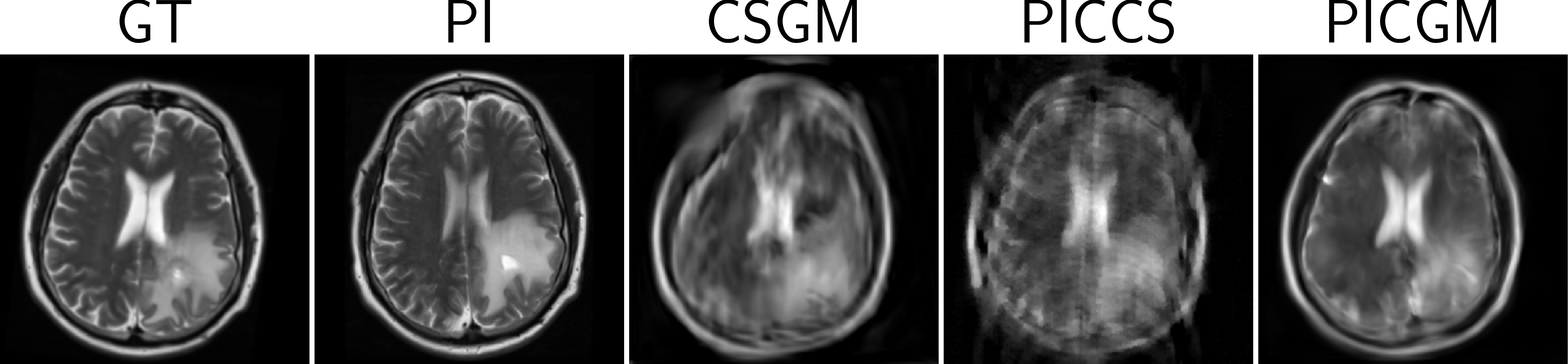}}
\caption{Ground truth (GT), prior image (PI) and image estimates from 8 fold undersampled measurements for misalignment configuration $C_4$}
\label{fig:robustness_fig}
\end{center}
\vspace{-10pt}
% \vskip -0.2in
\end{figure}

% discussion

\section{Discussion and Conclusion}\label{sec:concl}
Firstly, as seen in the numerical studies presented here, the object to-be-imaged and the prior image can be similar in some ways, while differing in others, due to factors such as evolution of the object over time, changes in the exact configuration of the imaging instrument, the depth at which the tomogram slice is recorded, changes in the defining features due to the tumor, {changes in the contrast mechanism used, and small changes in pose and orientation.
Such changes cannot always be effectively modeled by use of the sparse-differences model assumed in the case of PICCS. PICGM is able to outperform the other methods examined due to the fact that the StyleGAN gives control over individual semantic features in an image via the disentangled latent space.}

An important point that needs to be addressed is the choice of $p_1, p_2$ which determines which styles are kept fixed to those of the prior image and which styles are optimized over. Since we require $p_1, p_2$ to be multiples of $k$, the maximum choices for each of $p_1$ and $p_2$ are limited to $L$, the number of layers in the synthesis network. Moreover, in several realistic scenarios, the nature of the semantic differences between the ground truth and the prior image are known, and can be directly related to specific style vectors by performing style-mixing experiments on the trained StyleGAN. For example, small tumors correspond to the medium-scale features, and can be controlled by using the corresponding style vectors. When recovering an MRI image corresponding to one contrast mechanism with the help of a prior image from a different contrast mechanism, the style vectors corresponding to the variation in contrast are known beforehand. In our numerical studies, the values of $p_1, p_2$ as well as the regularization parameter $\lambda$ were kept fixed over the test dataset. Since the dataset may not be homogeneous with respect to the differences between the prior image and ground truth, further investigation is needed to devise an adaptive strategy for picking $p_1, p_2$. Several techniques for automatically picking regularization parameters exist in the compressive sensing literature which could be adapted to the proposed method \cite{refmri, autoreg}. In the presented numerical studies, the best parameters were selected by tuning over 4-5 $(p_1, p_2)$ pairs. The styles can be grouped into coarse, medium and fine for choices of $p_i$ between 0-2$k$, 2$k$-8$k$ and 8$k$-14$k$ respectively. As long as $p_i$ resided in the correct group, a $\sim$10\% median degradation in RMSE across the incorrect choices of $(p_1, p_2)$ was encountered.

The theoretical analysis presented here potentially offers a way to analyze the stability of the reconstructed estimate. In particular, the presented analysis applies to non-uniform recovery of in-distribution objects, that are sufficiently regular, in the sense that the Jacobian of the generative network computed at their latent representation is well behaved. This means that small changes in the estimated $\hat{\w}$ due to the measurement noise do not correspond to huge changes in the estimated image, which is what is obtained via the presented recovery guarantees. Hence, if the proposed algorithm converges to a solution that does not satisfy the properties presented in \autoref{lem:ind_properties}, it likely has converged to a highly unstable and nonlinear region of the optimization landscape. This provides a certain degree of graceful failure, in the face of hallucinations that may occur in the image reconstruction \cite{hallucinations}. {It must be noted that certain assumptions in the theoretical analysis, such as the validity of a useful S-REC for non-Gaussian sensing matrices, and obtaining near-optimal approximations to \autoref{eqn:picgm} are not completely realistic, but have been made in recent literature, and there has been progress towards addressing them \cite{admm_csgm, ilo}}.

Lastly, a task-informed approach to evaluating image reconstruction algorithms is necessary in addition to traditional image quality metrics employed in this study. 

\section*{Acknowledgements}
The authors would like to thank Sayantan Bhadra and Weimin Zhou for their help. This work was supported in part by NIH Awards EB020604, EB023045, NS102213, EB028652, and NSF Award DMS1614305.

% Acknowledgements are not included in this version.


\bibliography{picgm}
\bibliographystyle{icml2021}

\onecolumn
\appendix
\section{Theoretical Analysis}
As described in the main manuscript, the theoretical analysis presented here provides a non-uniform recovery guarantee, which applies to \textit{in-distribution objects} that are in the range of the StyleGAN $G$, further constrained by styles from the prior image. 
In contrast to the theoretical results presented in \cite{bora} where the Lipschitz constant of the generator network $G$ is used to bound the number of measurements, 
the analysis presented here is in terms of the expected Frobenius norm of its Jacobian.
Due to this, the presented analysis applies to generative networks having Lipschitz constants that are large as compared to the typical scaling of differences under the network, or generative networks that are not Lipschitz stable, such as StyleGAN2. The price paid is in terms of the guarantee being non-uniform, and allowing for an additional network-dependent term in the reconstruction error. Nevertheless, as we show, the proposed guarantees are useful in analyzing the behaviour of generative model-constrained reconstruction in general, and the PICGM method in particular. 

We begin by defining the notation used.
\begin{notation}
{$~$}\\
\vspace{-20pt}
\begin{enumerate}
    \item Let $p_{\w}$ denote the distribution of the extended latent space vector $\w = [\u_1^\top ~ \u_2^\top ~\dots ~ \u_L^\top]^\top \in \Wp$, where $\u_i = g_{\rm mapping}(\z_i)$, $\z_i \sim \mathcal{N}(0,I_k)$, $\z_i$'s are independently distributed, with $I_k$ denoting the real $k\times k$ identity matrix.
    \item Recall that as evidenced by \cite{stylegan_inv_gaussian}, it can be assumed that if $\w \sim p_\w$, $\v = \text{LReL}_\alpha(\w)$ is distributed as a multivariate Gaussian distribution. Let $\bar{\v}$ and $\Sigma$ be its mean and covariance matrix respectively. Recall that $\text{LReL}_\alpha$ denotes the leaky-ReLU nonlinear activation, defined as
    \begin{align}
            \text{LReL}_{\alpha}(\x)_i = \left\{\begin{matrix} x_i, \quad x_i \geq 0, \\
                                                \alpha x_i, \quad x_i < 0.
    \end{matrix}\right.
    \end{align}
    The value of $\alpha$ is the reciprocal of the scaling value for negative numbers included in the last leakyReLU layer in the mapping network $g_{\rm{mapping}}$.
    \item Let $p_1, p_2$ be positive integer multiples of $k$, with $1 \leq p_1 < p_2 \leq K$. Let $P = p_2 - p_1$. Let $\mathcal{W}^+_{p_1,p_2}$ be the $P$-dimensional subspace of $\mathcal{W}^+$ containing all $\w$ such that $\w_{1:p_1} = \mathbf{0}, \w_{p_2:K} = \mathbf{0}$.
    \item Let
    \begin{align*}
        B_\w^K (r) &:= \left\{ \w ~ s.t. ~ \norm{\mathrm{LReL}_\alpha(\w) - \bar{\v}}_\Sigma \leq r\right\},\\
        B_\v^K (r) &:= \left\{ \v ~ s.t. ~ \norm{\v - \bar{\v}}_\Sigma \leq r\right\}.\\
        \intertext{Similarly, let}
        B_\w^{p_1,p_2} (r) &:= \Big{\{}\w ~ s.t. ~ \w \in B_\w^K (r), \w_{1:p_1}=\wpi_{1:p_1}, \w_{p_2:K} = \wpi_{p_2:K}\Big{\}},\\
        B_\v^{p_1,p_2} (r) &:= \Big{\{}\v ~ s.t. ~ \v \in B_\v^K (r), \v_{1:p_1}=\v^{\rm(PI)}_{1:p_1}, \v_{p_2:K} = \v^{\rm(PI)}_{p_2:K}\Big{\}},
    \end{align*}
    where for a vector $\x \in \R^{K}$, $\norm{\x}_\Sigma^2 := \x^\top \Sigma^{-1}\x$. Note that $\alpha > 1$ and $1 \leq p_1 < p_2 \leq K$.
    
    \item Let $J(\w)$ denote the Jacobian of the synthesis network $G$ evaluated at $\w$. Let $J_{p_1:p_2}(\w)$ denote the Jacobian of $G$ with respect to $\w_{p_1:p_2}$ , evaluated at $\w$.
\end{enumerate}

\end{notation}

We first prove the following series of lemmas.

\begin{lemma}[]\label{lem:enet_Bk}

Let $\sigma_1 \geq \sigma_2 \geq \dots \geq \sigma_K$ be the singular values of $\sqrt{\Sigma}$.
Let ${\bm\sigma} = [\sigma_1 ~ \sigma_2 ~ \dots ~ \sigma_K]^\top$.
For $r > 1$, if $\mathcal{N}^{p_1,p_2}_{\w}(\epsilon)$ is an optimal $\epsilon$-net of $B^{p_1,p_2}_\w(r)$, then
\begin{align*}
    \log | \mathcal{N}^{p_1,p_2}_\w | \leq P\log\left[ \frac{6r}{\epsilon}\left(\epsilon + \frac{\norm{\bm\sigma}_2}{\sqrt{K}}\right) \right].
\end{align*}
\end{lemma}
\begin{proof}

First, note that if $\w \sim p_\w$, then the subsections of $\w$ corresponding to the different style inputs, i.e. $\w_{lk+1:(l+1)k}, ~ l=0,\dots, L-1$ are distributed such that $\w_{lk+1:(l+1)k}$ is independent to and identically distributed as $\w_{l'k+1:(l'+1)k}$ if $l\neq l'$. This implies that the singular values of $\sqrt{\Sigma}$ are degenerate to a certain degree. Specifically, 
\begin{align}\label{eqn:degeneracy}
    \sigma_{jL+1} = \sigma_{jL+2} = \dots = \sigma_{(j+1)L}, \quad j = 0,2,\dots,k-1
\end{align}

Let $\sigma_1' \geq \sigma_2' \geq \dots \geq \sigma_P'$ be the singular values of $\sqrt{{\rm Cov}(\v_{p_1:p_2})}$, and let $\bm{\sigma}' = [\sigma_1' ~ \sigma_2' ~ \dots ~\sigma_P']^T$. Therefore, by \autoref{eqn:degeneracy}, 
\begin{align}\label{eqn:degen2}
    \frac{\norm{\bm \sigma'}}{\sqrt{P}} = \frac{\norm{\bm \sigma}}{\sqrt{K}}.
\end{align}

Note that $B^{p_1,p_2}_\v(r)$ is an ellipsoid with center $\bar{\v}_{p_1:p_2}$ and principal radii of lengths $\sigma_i' r, ~ i = 1,2,\dots, P$. Assume for a moment, that there exists an integer $p$ such that $\sigma'_p r > \epsilon \geq \sigma'_{p+1} r$. Let $\mathcal{N}_\v (\epsilon)$ be an optimum $\epsilon$-net of $B^{p_1,p_2}_\v(r)$.

Therefore, by Theorem 2 in \cite{ellipsoid_cover},
\begin{align*}
    \log |\mathcal{N}_\v (\epsilon)| &\leq \sum_{i=1}^p \log \left( \frac{r\sigma_i'}{\epsilon} \right) + P\log 6,\\
    &\leq \log\left[\left(\frac{r}{\epsilon}\right)^P\prod_{i=1}^p \sigma_i' \prod_{i=p+1}^P \epsilon \right] + P\log 6, \tag{since $r > 1$,}\\\\
    &\leq P\log\left[ \frac{6r}{\epsilon}\left(\epsilon + \frac{1}{P}\sum_{i=1}^P{\sigma_i'}\right) \right], \tag{by AM-GM inequality,}\\
    &\leq P\log\left[ \frac{6r}{\epsilon}\left(\epsilon + \frac{\norm{\bm\sigma'}}{\sqrt{P}}\right) \right], \tag{by AM-RMS inequality,}\\
    &\leq P\log\left[ \frac{6r}{\epsilon}\left(\epsilon + \frac{\norm{\bm\sigma}}{\sqrt{K}}\right) \right].
\end{align*}

Observe that this bound is valid even if $\epsilon > \sigma'_1r$ or $\epsilon < \sigma'_Pr$.
Since $\alpha > 1$, $\text{LReL}_{1/\alpha}(.)$ is a bijective function with Lipschitz constant 1. Therefore, for every $\v_1 = \text{LReL}_\alpha(\w_1)$, and $\v_2 = \text{LReL}_\alpha(\w_2)$,
\begin{align*}
    \norm{\w_1 - \w_2} &\leq \norm{\v_1 - \v_2}.\\
    \intertext{Therefore,}
    \log |\mathcal{N}^{p_1,p_2}_\w (\epsilon) | &\leq P\log\left[ \frac{6r}{\epsilon}\left(\epsilon + \frac{\norm{\bm\sigma}}{\sqrt{K}}\right) \right]. \numberthis{}
\end{align*}
\end{proof}

Now, the following assumptions about StyleGAN are made.

\begin{enumerate}
    \item \textbf{Path length regularity}:
    \begin{align*}
    \mathbb{E}_{\w \sim p_\w} \big( \norm{J(\w)}_F - a \big)^2 < b,\tag{AS1}\label{eqn:as1}
    \end{align*}
    where $b > 0$ and $a = \mathbb{E}_\w \norm{J(\w)}_F$ are global constants. As described in the main manuscript, this assumption is inspired by the path-length regularization used in \cite{stylegan2}. Although during training $a$ is implemented as $\mathbb{E}_\w \mathbb{E}_{\y \sim \mathcal{N}(0,I_n)} \norm{J(\w)^\top\y}$, as per the Hanson-Wright inequality, this concentrates to $\mathbb{E}_\w \norm{J(\w)}_F$, when $\y$ is high-dimensional \cite{hdp_book}. The value of $a$ was estimated by empirically estimating $\mathbb{E}_{\w \sim p_\w} \mathbb{E}_{\y \sim \mathcal{N}(0,I_n)} \norm{J(\w)^\top\y}$ using 100 samples $\w$ drawn from $p_\w$ and 100 samples $\y$ drawn from $\mathcal{N}(0,I_n)$ for each sample of $\w$. $b$ was empirically estimated over the same dataset of samples using \autoref{eqn:as1}. The values of $a$ and $\sqrt{b}$ were estimated to be around 80.1 and 16.7 respectively for the specific StyleGAN2 trained and used in the inverse-crime study.
    
    \item \textbf{Approximate local linearity}:
    \begin{align*}
        \mathbb{E}_{\w \sim p_\w} &\max_{\begin{smallmatrix}\w'\\ \norm{\w'-\w}\leq \epsilon\end{smallmatrix}} \mathcal{L}(\w', \w) \leq \beta^2(\epsilon),\tag{AS2}\label{eqn:local_lin}\\
        \intertext{where}
        \mathcal{L}(\w',\w) &= \left\|G(\w') - G(\w) - J(\w)(\w' - \w) \right\|_2^2 
    \end{align*}
    This property essentially measures how close $G$ is to its linear approximation in an $\epsilon$-neighborhood around a point $\w$. For ease of notation, we will write 
    \begin{equation}
    \phi_{p_1,p_2}^2(\epsilon; \w) := \max_{\begin{smallmatrix}\w'\\ \norm{\w'-\w}\leq \epsilon \\ \w'-\w \in \mathcal{W}^+_{p_1,p_2} \end{smallmatrix}} \mathcal{L}(\w', \w),
    \end{equation}
    with $\phi_{p_1,p_2}(\epsilon; \w) \geq 0$. 
    Approximate estimates of $\beta^2(\epsilon)$ were obtained for several values of $\epsilon$ by first computing the Jacobian at a point $\w \sim p_\w$, and then iteratively maximizing $\mathcal{L}(\w',\w)$ using a projected gradient ascent-type algorithm.
    Figure \ref{fig:beta_plot} shows the plot of $\beta^2(\epsilon)$ versus $\epsilon$ estimated over a dataset of 100 samples $\w$ from $p_\w$ for the StyleGAN2 trained and used in the inverse-crime study.
    
\end{enumerate}

\textbf{Lemma 4.1.}
If $\w$ is a sample from $p_\w$, then it satisfies the following three properties with probability at least $1 - O(1/K)$:
\begin{align*}
\norm{J(\w)}_F &\leq \sqrt{K}a, \tag{P1}\label{eqn:p1}\\
\phi_{1,K}(\epsilon; \w) &\leq \sqrt{K}\beta(\epsilon)\tag{P2}\label{eqn:p2}\\
\norm{\Sigma^{-1/2}\mathrm{LReL}_{\alpha}(\w)}_2 &\leq \sqrt{K}(1+o(1))\tag{P3}\label{eqn:p3}
\end{align*}

\begin{figure}[ht]
\vskip 0.2in
\begin{center}
\centerline{\includegraphics[width=0.5\linewidth]{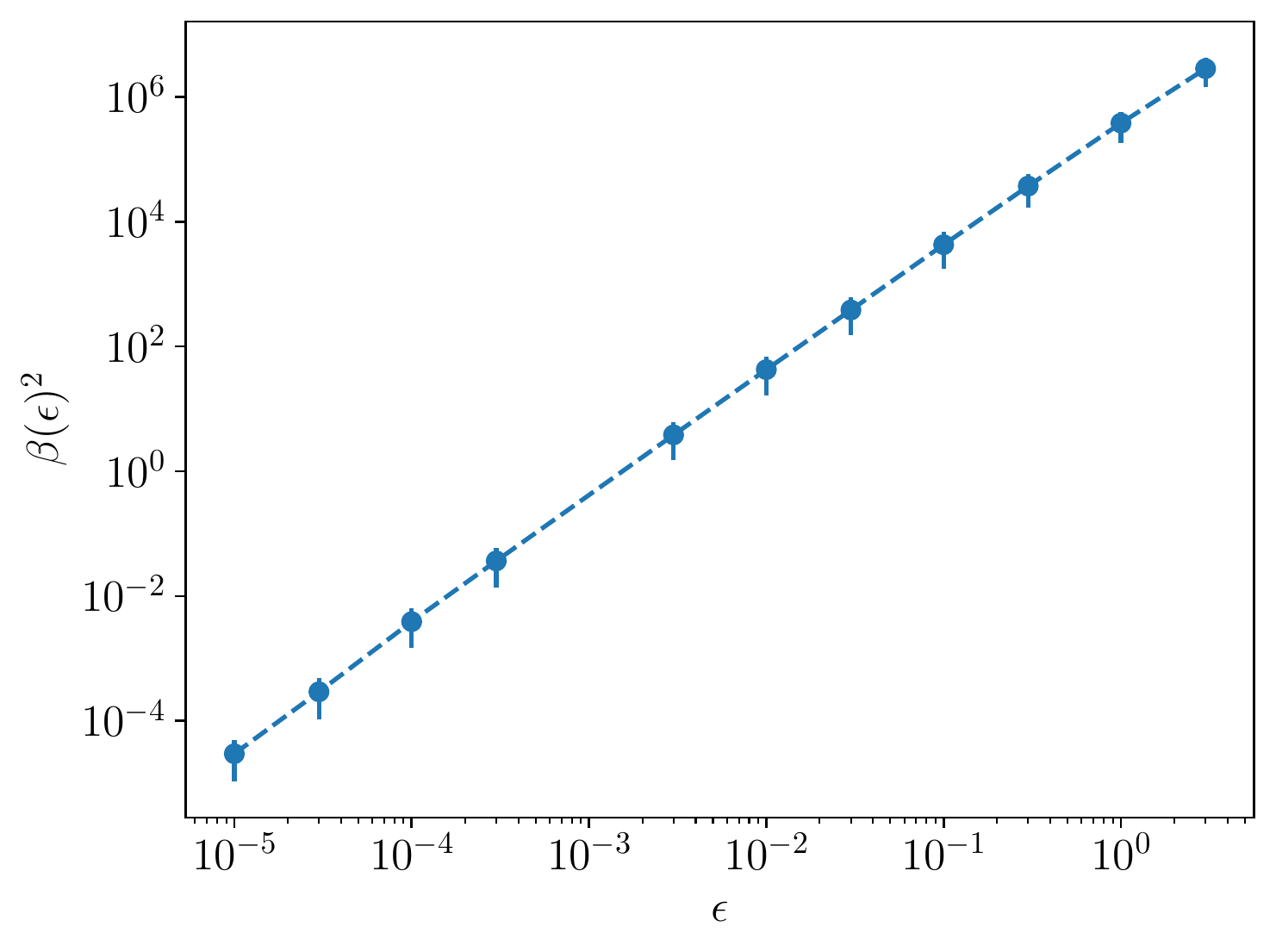}}
\caption{Plot of $\beta(\epsilon)^2$ versus $\epsilon$}
\label{fig:beta_plot}
\end{center}
\vskip -0.2in
\end{figure}

\begin{proof}
For $\w$ sampled from $p_\w$, \ref{eqn:p1} and \ref{eqn:p2} are true with probability at least $1 - \frac{a^2 + b}{Ka^2}$ and $1-1/K$ respectively, due to Markov's inequality \cite{hdp_book}. \ref{eqn:p3} is true with probability at least $1-\Omega(e^{-cK})$ due to concentration of norm \cite{hdp_book}. Therefore, by union bound, $\w$ satisfies all the three with probability at least $1 - O(1/K)$.
\end{proof}

As a consequence of Lemma 4.1, for $$\w_s = [\wpi_{1:p_1}{}^\top \quad \w_{p_1:p_2}^\top \quad \wpi_{p_2:K}{}^\top]^\top,$$ the following properties also hold with probability at least $1 - O(1/K)$ if $\w,\wpi \sim p_{\w}$:
\begin{align*}
    \norm{J_{p_1:p_2}(\w_s)}_F &\leq \sqrt{K}a, \tag{TP1}\label{eqn:tp1}\\
    \tag{Since $J_{p_1:p_2}(\w_s)$ is a sub-matrix of $J(\w_s)$}\\
    \phi_{p_1,p_2}(\epsilon; \w_s) &\leq \sqrt{K}\beta(\epsilon)\tag{TP2}\label{eqn:tp2}
\end{align*}

If $\w$ satisfies properties \ref{eqn:p1}, \ref{eqn:p2} and \ref{eqn:p3}, then $G(\w)$ will be referred to as an in-distribution image in the range of $G$, since these are the properties of a typical sample from the StyleGAN.

\begin{notation}
{$~$}\\
\vspace{-17pt}

Let $\tilde{B}^{p_1,p_2}_\w(r)$ be the set of all points in $B^{p_1,p_2}_\w(r)$ satisfying properties P1 and P2.
\end{notation}

\begin{lemma}\label{lem:covering_full}
Let $0 < \delta \leq \frac{a\norm{\bm\sigma}_2}{72\sqrt{K}}$. Let $\tilde{\mathcal{N}}_\f(\delta)$ be a discrete set in $G(\tilde{B}^{p_1,p_2}_\w(r))$ such that for every $\f \in G(\tilde{B}^{p_1,p_2}_\w(r))$, there exists an $\f_0 \in \mathcal{N}_\f(\delta)$ such that
\begin{align}
    \norm{\f - \f_0}_2 &\leq \delta + \sqrt{K}\beta(\delta/a).\\
    \intertext{Then,}
    |\tilde{\mathcal{N}}_\f(\delta) | &\leq P\log \left(\frac{145ar\norm{\bm\sigma}}{\sqrt{P}\delta}\right).
\end{align}
\end{lemma}

\begin{proof}
The outline of this proof is as follows. First, an $\delta/a$-covering over $\tilde{B}^{p_1,p_2}_\w(r)$ will be constructed. Then, each of the spherical balls covering $\tilde{B}^{p_1,p_2}_\w(r)$ is transformed into approximately an ellipsoid depending upon the Jacobian of $G$ at the center of the spherical ball. Then, each of these ellipsoids will be approximately covered by a $\delta$-net. The collection of all such approximate $\delta$-nets covering the individual ellipsoids will give an approximate $\delta$-net over $G(\tilde{B}^{p_1,p_2}_\w(r))$, which is the result required.

Let $\epsilon = \delta/a$. Let $\tilde{\mathcal{N}}_\w(\epsilon)$ be an optimal $\epsilon$-covering of $\tilde{B}^{p_1,p_2}_\w(r)$. Also, let $\mathcal{N}_\w(\epsilon)$ denote an optimal $\epsilon$-covering of $B^{p_1,p_2}_\w(r)$. Therefore, since $\tilde{B}^{p_1,p_2}_\w(r) \subseteq B^{p_1,p_2}_\w(r)$, \cite{hdp_book},
\begin{align*}
    \log |\tilde{\mathcal{N}}_\w(\epsilon)| &\leq \log|\mathcal{N}_\w(\epsilon/2)|\\
    &\leq P\log\left[ \left(\frac{12r}{\epsilon}\right)\left( \frac{\norm{\bm\sigma}}{\sqrt{K}} + \frac{\epsilon}{2} \right) \right]. \tag{using \autoref{lem:enet_Bk}.}
\end{align*}

Now, consider a point $\w_0 \in \tilde{\mathcal{N}}_{\w}(\epsilon)$. Therefore, 
%by Property \ref{eqn:tp2}, 
for every $\w' \in \tilde{B}_{\w}^{p_1,p_2}(r)$ such that $\norm{\w' - \w_0} \leq \epsilon$, we have 
\begin{align}
    \norm{G(\w') - G(\w_0)} \leq \norm{J(\w_0)(\w'-\w_0)} + \sqrt{K}\beta(\epsilon).
\end{align}

Therefore, upto an error of $\sqrt{K}\beta(\epsilon)$, $G(\w') - G(\w_0)$ lies in an ellipsoid $E(\w_0)$ with principal radii $\varsigma_1\epsilon, \varsigma_2\epsilon, \dots, \varsigma_P\epsilon$, where $\varsigma_1 \geq \varsigma_2 \geq \dots \geq \varsigma_P$ are the singular values of $J_{p_1:p_2}(\w_0)$.

Let $\mathcal{N}_\delta(\w_0)$ be a $\delta$-covering of $E(\w_0)$. For a moment assume that there exists an integer $p$ such that $\varsigma_p < a \leq \varsigma_{p+1}$. Therefore, from Theorem 2 in \cite{ellipsoid_cover}, we have
\begin{align*}
    \log|\mathcal{N}_\delta(\w_0)| &\leq \sum_{i=1}^p \log\left(\frac{\varsigma_i}{a}\right) + P\log 6,\\
    &= \log\left[\frac{1}{a^P}\prod_{i=1}^p\varsigma_i\prod_{i=p+1}^P a\right] + P\log 6\\
    &\leq P\log\left[\frac{1}{a}\left(\frac{\norm{\bm{\varsigma}}_2}{\sqrt{P}} + a\right)\right] + P\log 6,\\
    &\leq P\log\left[\frac{1}{a}\left(\frac{\norm{\bm{\varsigma}}_2}{\sqrt{P}} + a\sqrt{\frac{K}{P}}\right)\right] + P\log 6,\\
    &\leq P\log 12\sqrt{\frac{K}{P}}.\tag{using Property \ref{eqn:tp1}.}
\end{align*}
Note that this bound holds even when $a \geq \varsigma_1$, or $a \leq \varsigma_P$.

Therefore, for every $\w'$ such that $\norm{\w' - \w_0} \leq \epsilon$, there exists a point $\f_1$ in $\mathcal{N}_\delta(\w_0)$ such that
\begin{align*}
    &\norm{G(\w_0) + J(\w_0)(\w'-\w_0) - \f_1} \leq \delta,\\
    \Rightarrow&\left\|G(\w_0) + J(\w_0)(\w'-\w_0) - G(\w')~+  G(\w') - \f_1\right\|
    \leq \delta,\\
    \Rightarrow &\norm{G(\w') - \f_1} \leq \delta + \sqrt{K}\beta(\delta/a)\numberthis{}\\
    \tag{by triangle inequality}
\end{align*}

This holds for all $\w_0 \in \tilde{\mathcal{N}}^K_\w(r)$. Therefore, a suitable candidate set for $\tilde{\mathcal{N}}_\f(\delta)$ is 
\begin{align}
    \tilde{\mathcal{N}}_\f(\delta) = \{\w_1 ~ s.t. ~ \w_1 \in \mathcal{N}_\delta(\w_0), ~ \w_0 \in \tilde{\mathcal{N}}^K_\w(r) \}
\end{align}

Therefore, 
\begin{align*}
    \log |\tilde{\mathcal{N}}_\f(\delta)| &\leq P\log 12\sqrt{\frac{K}{P}} + P\log\left[ \frac{12ar}{\delta}\left( \frac{\norm{\bm\sigma}}{\sqrt{K}} + \frac{\delta}{2a} \right) \right],\\
    &\leq P\log\left[ \left(\frac{144ar\sqrt{K}}{\delta\sqrt{P}}\right)\left( \frac{\norm{\bm\sigma}}{\sqrt{K}} + \frac{\delta}{2a} \right) \right],\\
    &\leq P\log \left(\frac{145ar\norm{\bm\sigma}}{\sqrt{P}\delta}\right).
\end{align*}
\end{proof}

\textbf{Proof of Theorem 4.1}:

First, we note that the result of \autoref{lem:covering_full} applies to a decreasing sequence of $\delta$'s: $\delta_i = \delta_0 / 2^i$. Also, from \autoref{fig:beta_plot}, we note that $\beta(\epsilon)$ goes polynomially with $\epsilon$ with $\beta(0) = 0$. Due to these,
Lemma 8.2 in \cite{bora} can be reformulated as following, with the proof proceeding similarly as \cite{bora}.
\begin{lemma}\label{lem:small_scale_error}
Let $H \in \mathbb{R}^{m\times n}$ be a matrix with iid Gaussian random elements having mean 0 and variance $1/m$. Let $0 < \delta \leq \frac{a\norm{\bm\sigma}_2}{72\sqrt{K}}$. For all $\delta'<\delta$, let $\beta(\delta'/a)$ go polynomially as $\delta'/a$, with $\beta(0) = 0$. If
\begin{align}
    m = {\mathrm{\Omega}}\left(P \log\frac{ar\norm{\bm\sigma}}{\delta}\right), 
\end{align}
then for any $\f \in G(\tilde{B}^{p_1,p_2}_\w(r))$, if $\f' = \argmin_{\hat{x}\in \tilde{\mathcal{N}}_\f(\delta)} \norm{\f - \hat{\f}}$, $\norm{H(\f - \f')} \leq O(\delta + \sqrt{K}\beta(\delta/a))$ with probability $1 - e^{-{\rm\Omega}(m)}$.\\
\end{lemma}

The proof goes similarly as the proof of Lemma 8.2 in \cite{bora}. Furthermore, similar to Lemma 4.1 in \cite{bora}, \autoref{lem:covering_full} and \autoref{lem:small_scale_error} give rise to the set-restricted eigenvalue condition of $H$ on $G(\tilde{B}^{p_1,p_2}_\w(r))$ as follows:

\begin{lemma}[Set-restricted eigenvalue condition]\label{lem:app_srec}
Let $\tau < 1$. Let $H$ be an matrix with iid Gaussian-distributed elements with mean 0 and variance $1/m$. Let $0 < \delta \leq \frac{a\norm{\bm\sigma}_2}{72\sqrt{K}}$. For all $\delta'<\delta$, let $\beta(\delta'/a)$ go polynomially as $\delta'/a$, with $\beta(0) = 0$. If
\begin{align}
    m = {\rm\Omega}\left(\frac{P}{\tau^2}\log \frac{ar\norm{\bm\sigma}}{\delta}\right),
\end{align}
then $H$ satisfies the $\text{S-REC}\big{(}G(\tilde{B}^{p_1,p_2}_\w(r)), ~ 1-\tau, ~ \delta+\sqrt{K}\beta(\delta/a)\big{)}$ with probability $1 - e^{-{\rm\Omega}(\tau^2m)}$.
\end{lemma}

\autoref{lem:app_srec}, Lemma 4.3 in \cite{bora} and Lemma 4.1 imply
Theorem 4.1 which is restated here for convenience.\\

\newpage
\textbf{Theorem 4.1.} 

Let $H \in \mathbb{R}^{m\times n}$ satisfy $\textrm{S-REC}(G(\tilde{B}^{p_1,p_2}_\w(r)), \gamma, \delta+\sqrt{K}\beta(\delta/a))$. Let $\vec{n} \in \mathbb{R}^m$. Let $\w, \wpi \sim p_\w$. Let $\fpi = G(\wpi)$ be the known prior image. Let 
$$\tilde{\w} = [\wpi_{1:p_1}{}^\top \quad \w_{p_1:p_2}^\top \quad \wpi_{p_2:K}{}^\top]^\top.$$
Let $\tilde{\f} = G(\tilde{\w})$ represent the object to-be-imaged. Let $\g = H\tilde{\f} + \vec{n}$ be the imaging measurements. Let
\begin{align}
    \hat{\f} = \argmin_{\f\in G(\tilde{B}^K_\w(r))} \norm{\g - H\f}_2^2.
\end{align}
Then, 
\begin{align}
    \| \hat{\f} - \tilde{\f} \| \leq \frac{1}{\gamma} (2\norm{\vec{n}} + \delta + \sqrt{K}\beta(\delta/a))
\end{align}
with probability $1 - O(1/K)$.

\section{Additional Figures}
\subsection{Inverse crime study: $n/m = 5$}
\begin{figure*}[ht]
\vskip 0.2in
\begin{center}
\centerline{\includegraphics[width=\linewidth]{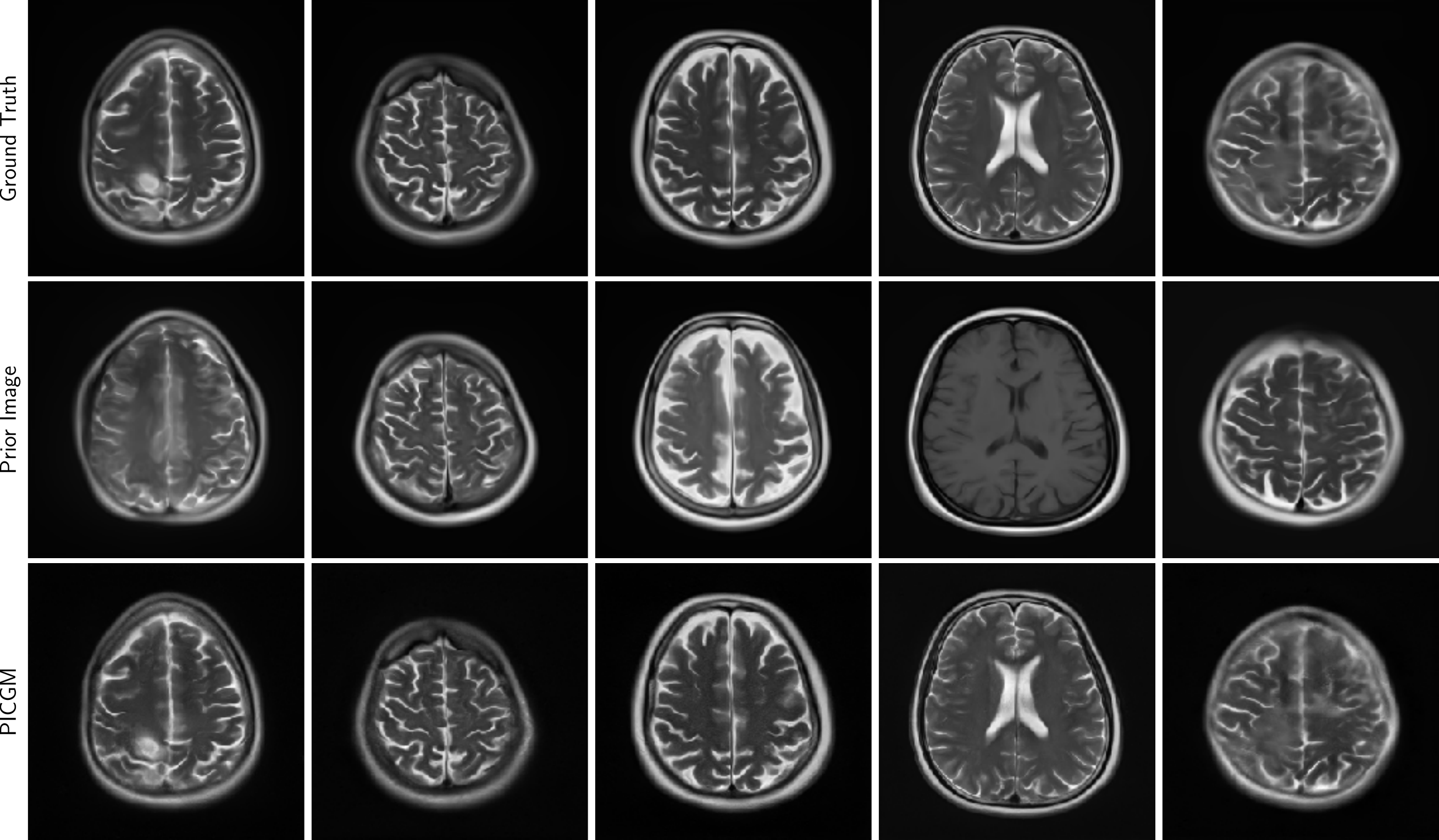}}
\caption{Ground truth, prior image and image estimated from Gaussian measurements with $n/m = 5$ using the proposed approach in the inverse crime case.}
\label{fig:inversecrime_images_gaussian_0.2}
\end{center}
\vskip -0.2in
\end{figure*}

\newpage
\subsection{Inverse crime study: $n/m = 10$}
\begin{figure*}[ht]
\vskip 0.2in
\begin{center}
\centerline{\includegraphics[width=\linewidth]{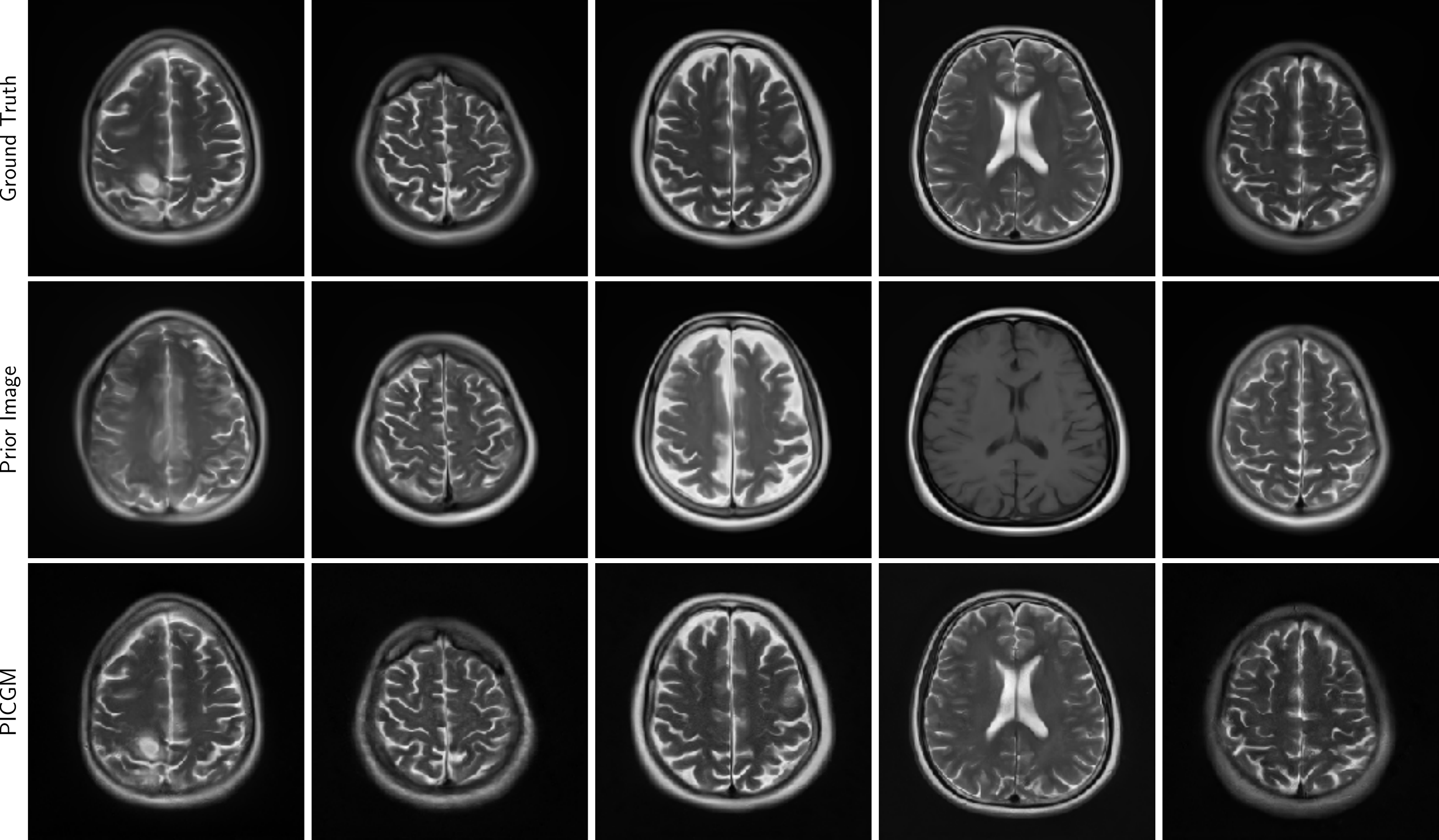}}
\caption{Ground truth, prior image and image estimated from Gaussian measurements with $n/m = 10$ using the proposed approach in the inverse crime case.}
\label{fig:inversecrime_images_gaussian_0.1}
\end{center}
\vskip -0.2in
\end{figure*}

\newpage
\subsection{Inverse crime study: $n/m = 20$}
\begin{figure*}[ht]
\vskip 0.2in
\begin{center}
\centerline{\includegraphics[width=\linewidth]{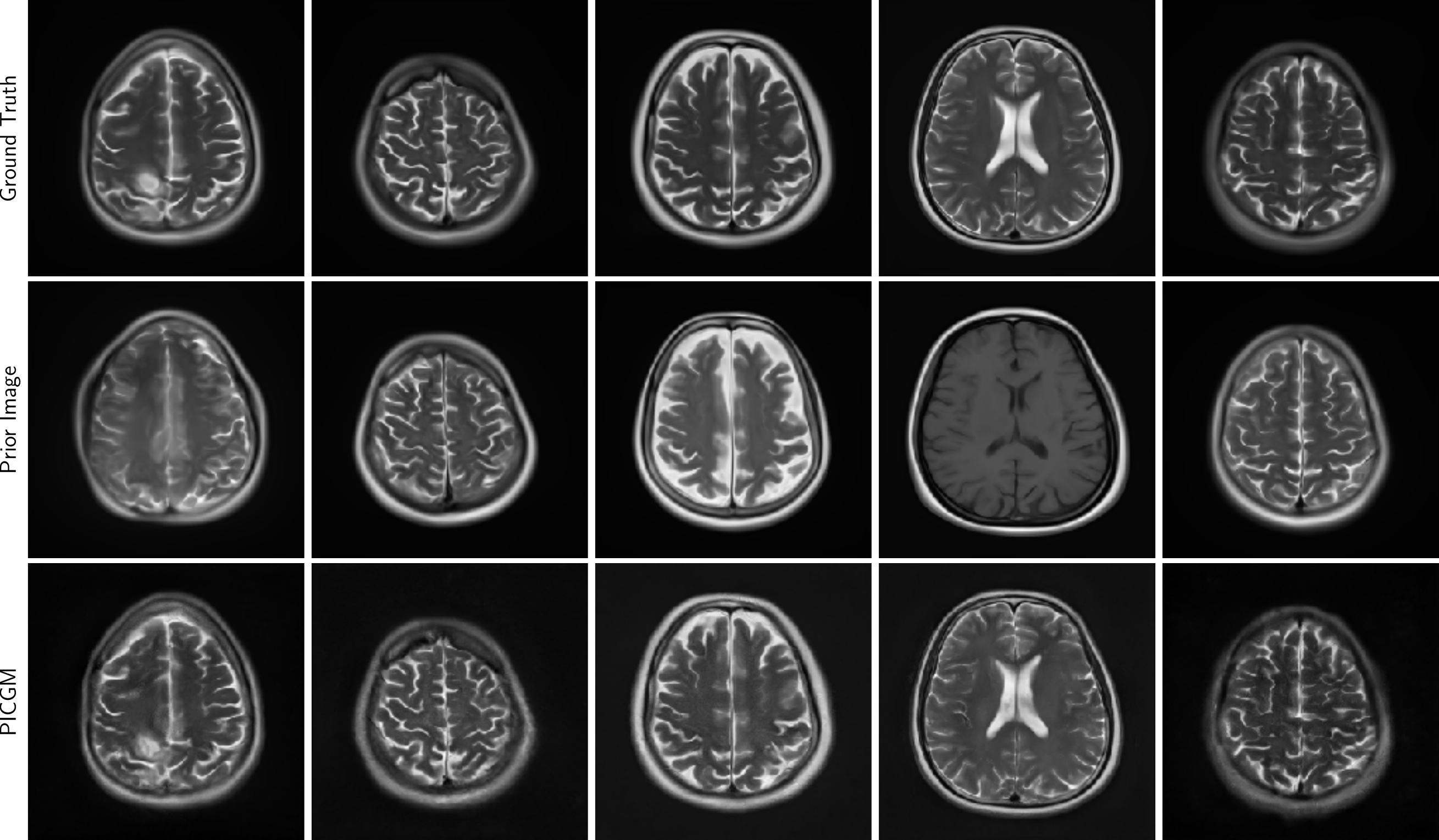}}
\caption{Ground truth, prior image and image estimated from Gaussian measurements with $n/m = 20$ using the proposed approach in the inverse crime case.}
\label{fig:inversecrime_images_gaussian_0.05}
\end{center}
\vskip -0.2in
\end{figure*}

\newpage
\subsection{Inverse crime study: $n/m = 50$}
\begin{figure*}[ht]
\vskip 0.2in
\begin{center}
\centerline{\includegraphics[width=\linewidth]{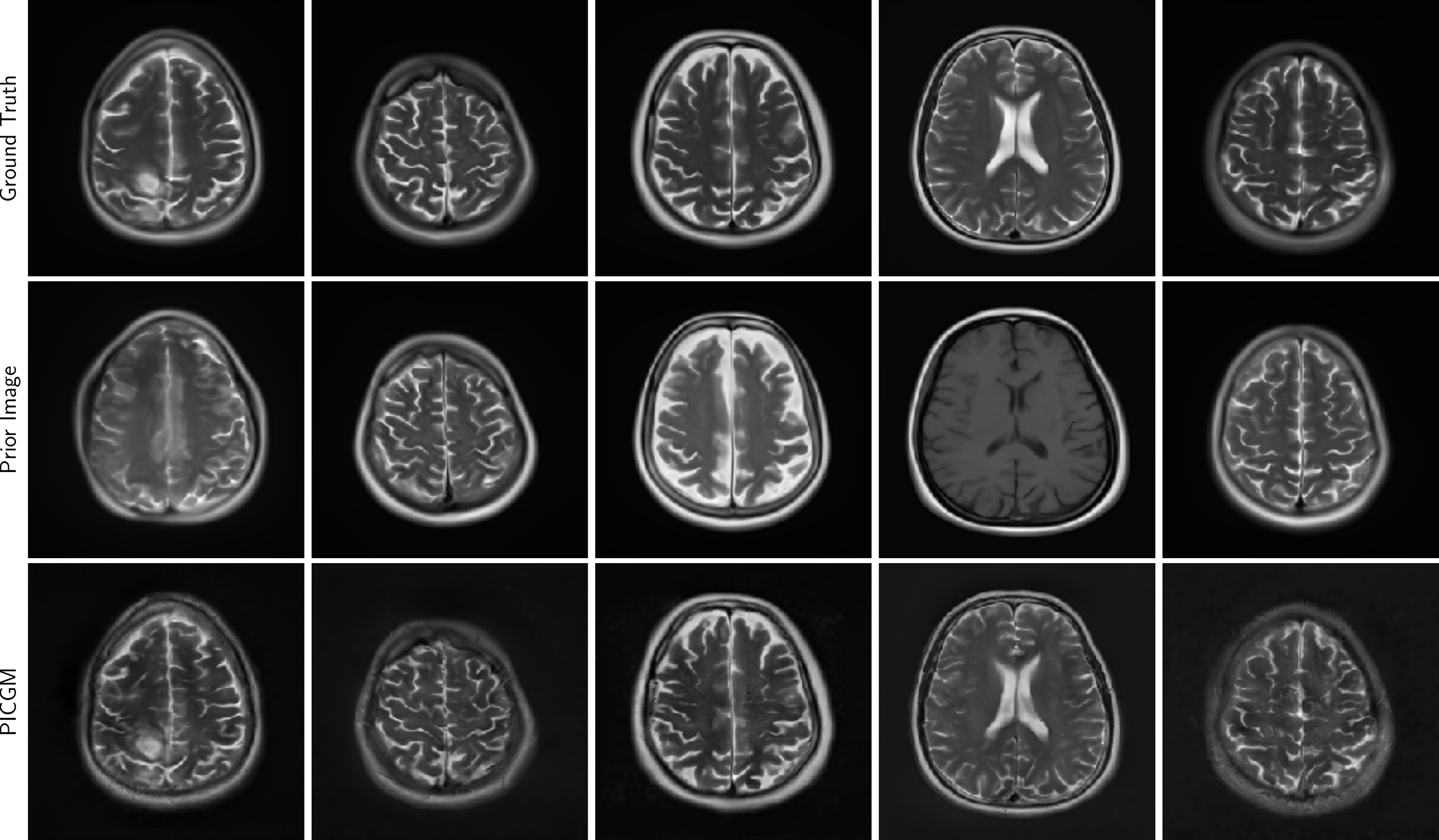}}
\caption{Ground truth, prior image and image estimated from Gaussian measurements with $n/m = 50$ using the proposed approach in the inverse crime case.}
\label{fig:inversecrime_images_gaussian_0.02}
\end{center}
\vskip -0.2in
\end{figure*}

\newpage
\subsection{Face image study: $n/m = 50$}
\begin{figure*}[ht]
\vskip 0.2in
\begin{center}
\centerline{\includegraphics[width=0.7\linewidth]{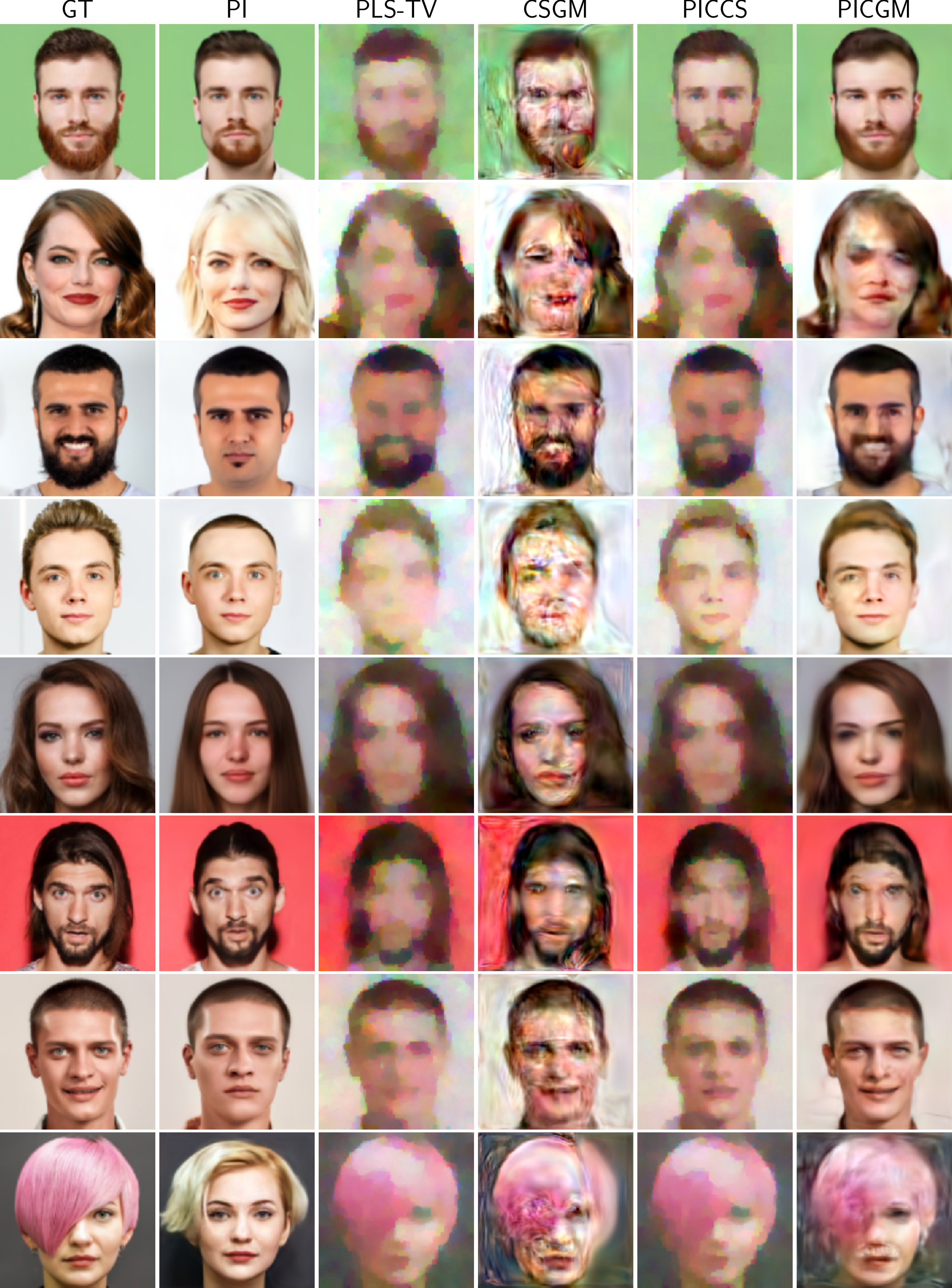}}
\caption{Ground truth, prior image and image estimated from Gaussian measurements with $n/m = 50$ using the proposed approach for the face image study.}
\label{fig:images_gaussian_0.02_many}
\end{center}
\vskip -0.2in
\end{figure*}

\newpage
\subsection{MR image study: $n/m = 2$}
\vskip 0.2in
\begin{center}
\centerline{\includegraphics[width=\linewidth]{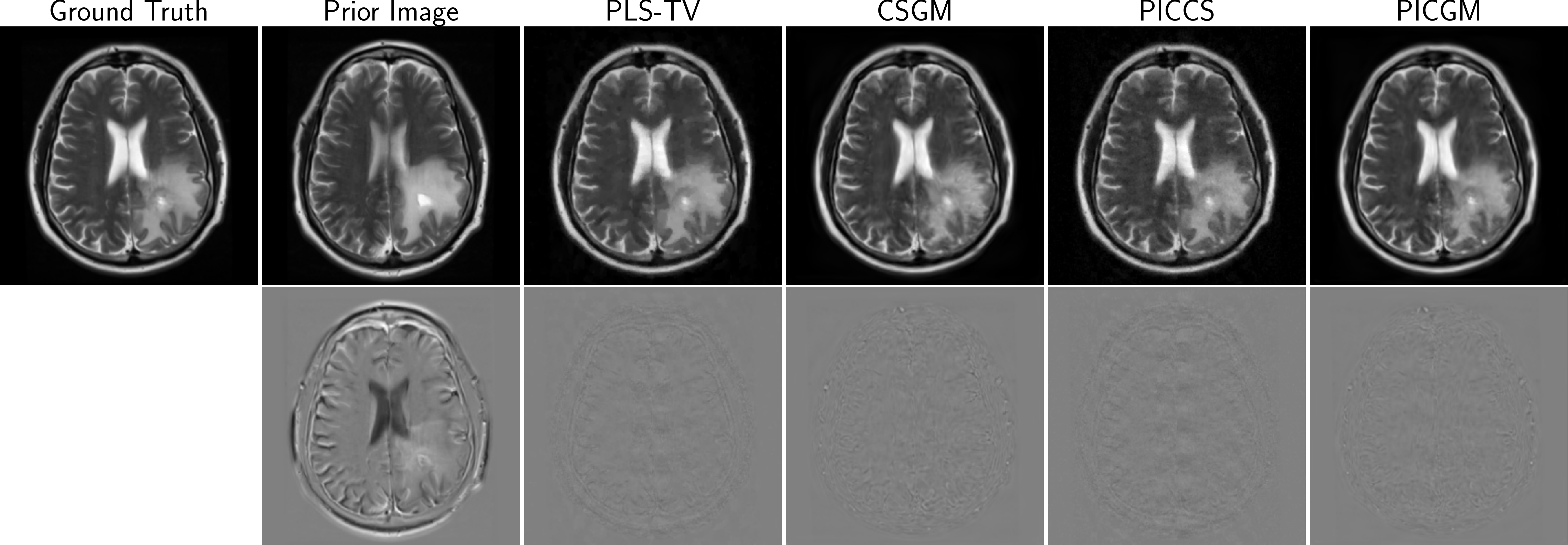}}
\vspace{10pt}
\centerline{\includegraphics[width=\linewidth]{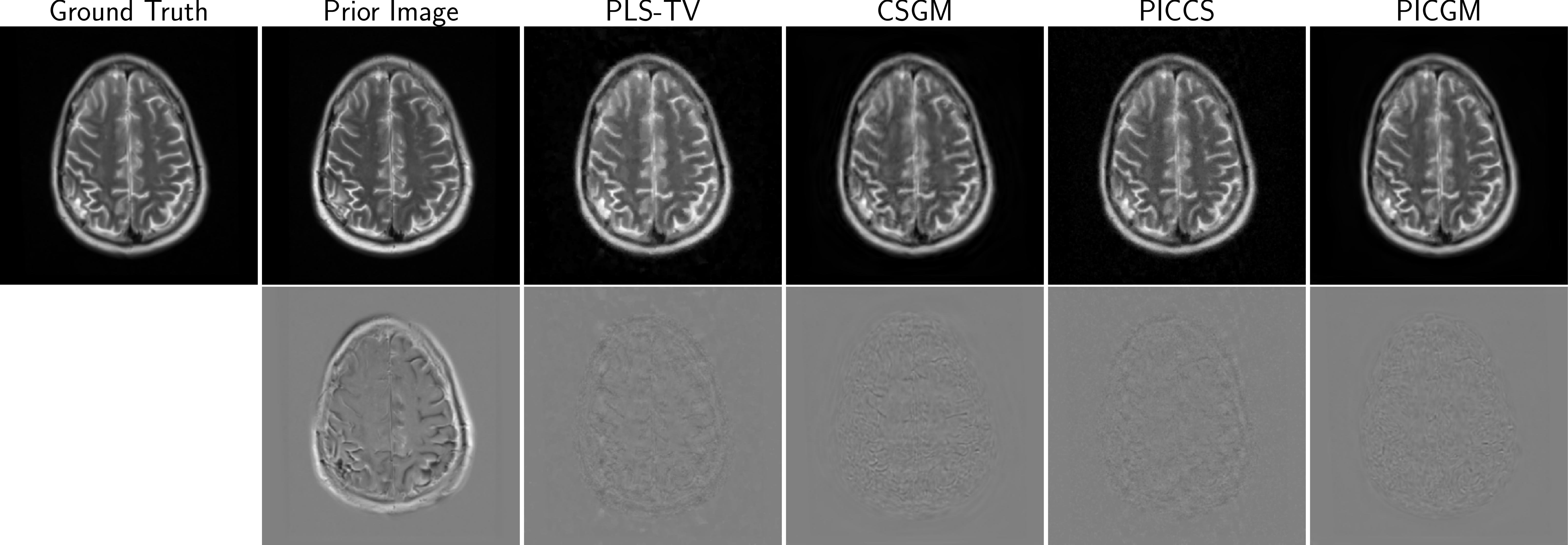}}
\vspace{10pt}
\centerline{\includegraphics[width=\linewidth]{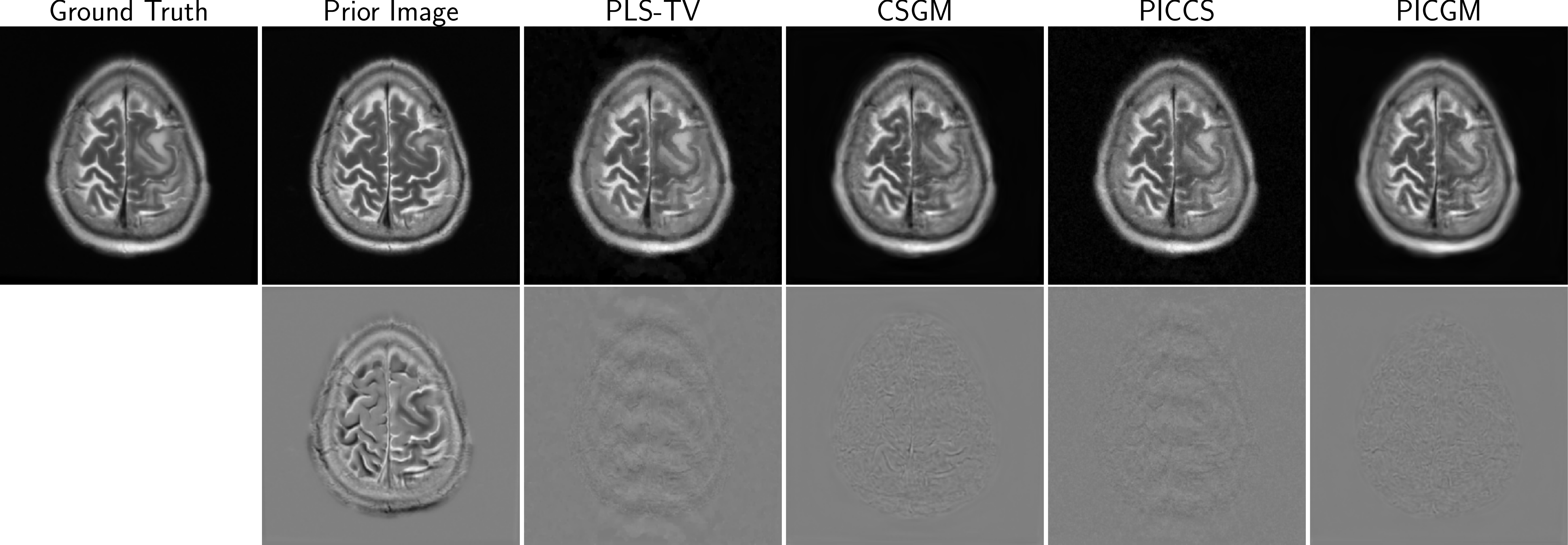}}
\captionof{figure}{Ground truth, prior image, and images reconstructed from simulated MRI measurements with $n/m = 2$ along with difference images for the MR image study}
\label{fig:recon_mask_rand_2x}
\end{center}
\vskip -0.2in

\newpage
\subsection{MR image study: $n/m = 4$}
\vskip 0.2in
\begin{center}
\centerline{\includegraphics[width=\linewidth]{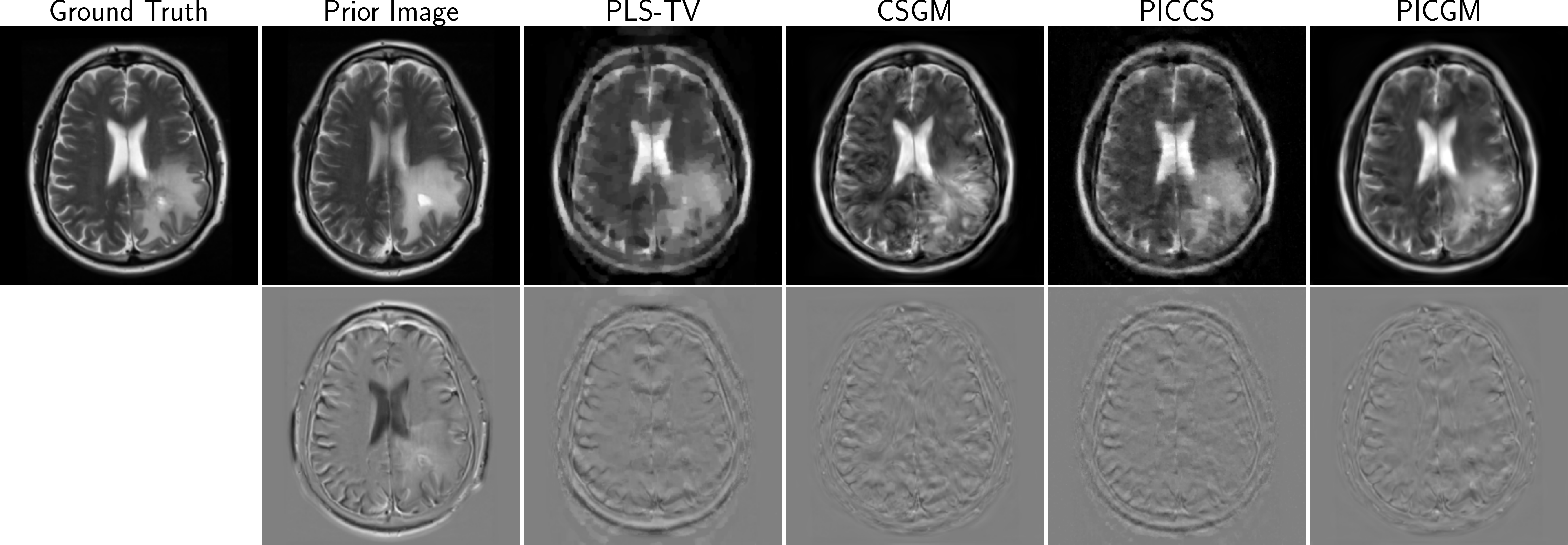}}
\vspace{10pt}
\centerline{\includegraphics[width=\linewidth]{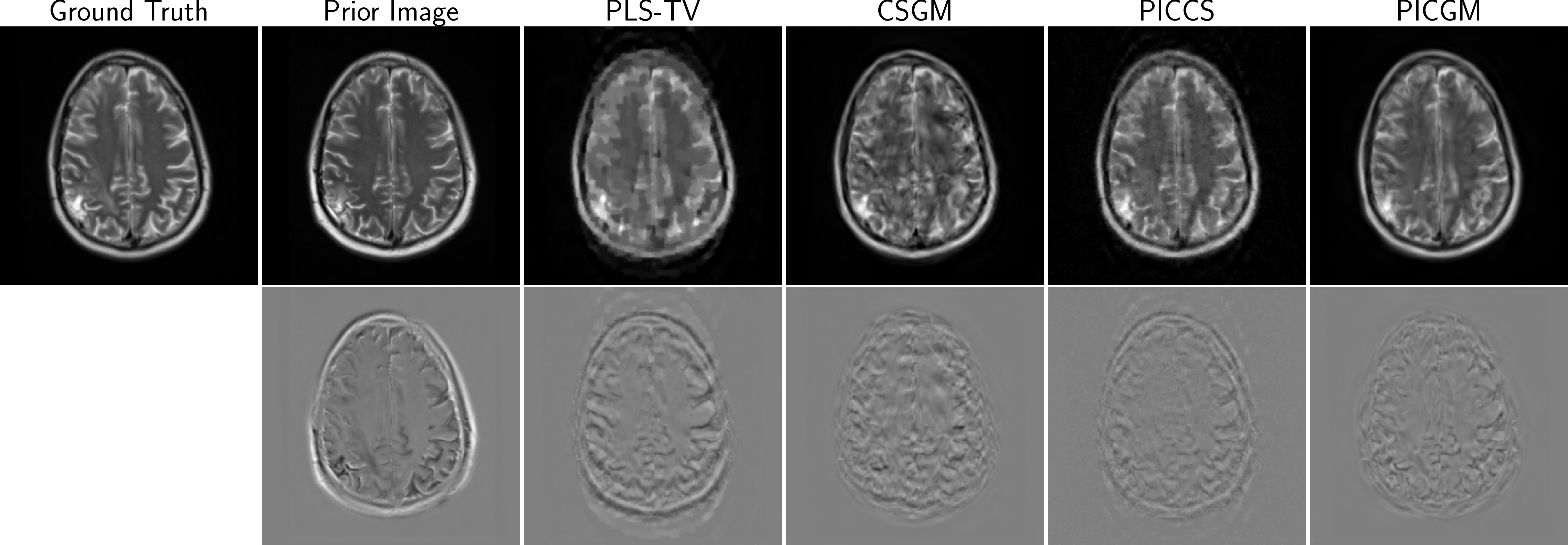}}
\vspace{10pt}
\centerline{\includegraphics[width=\linewidth]{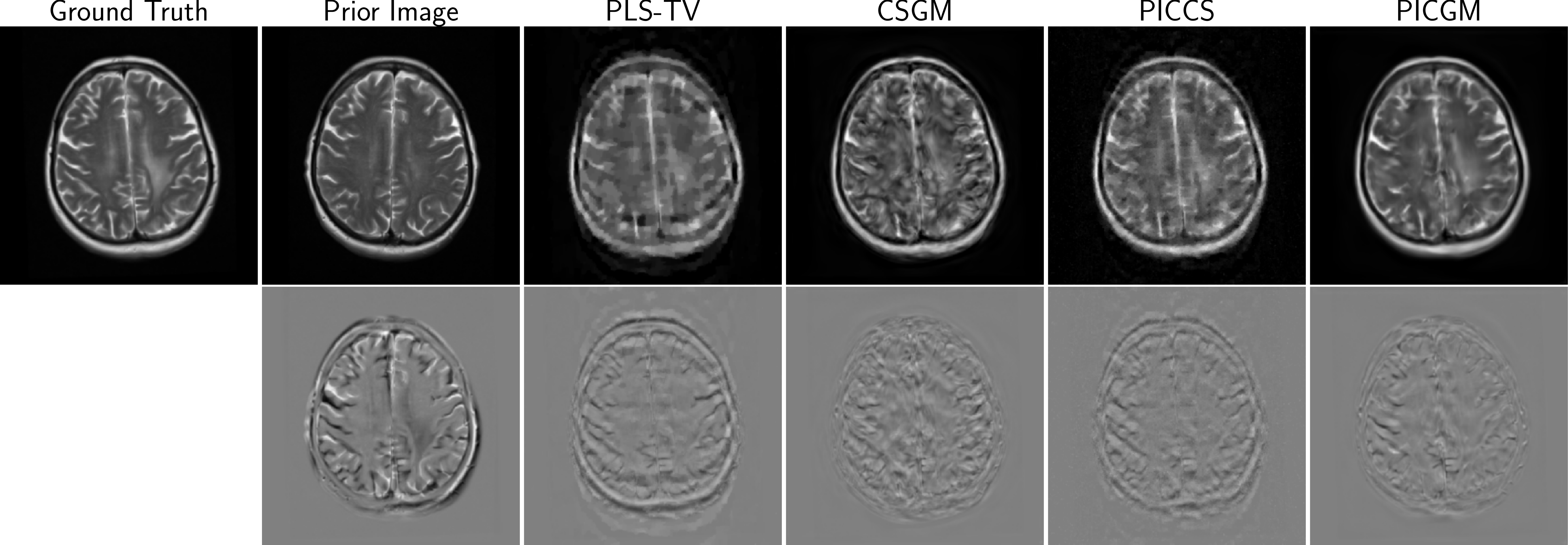}}
\captionof{figure}{Ground truth, prior image, and images reconstructed from simulated MRI measurements with $n/m = 4$ along with difference images for the MR image study}
\label{fig:recon_mask_rand_4x}
\end{center}
\vskip -0.2in

\newpage
\subsection{MR image study: $n/m = 6$}
\vskip 0.2in
\begin{center}
\centerline{\includegraphics[width=\linewidth]{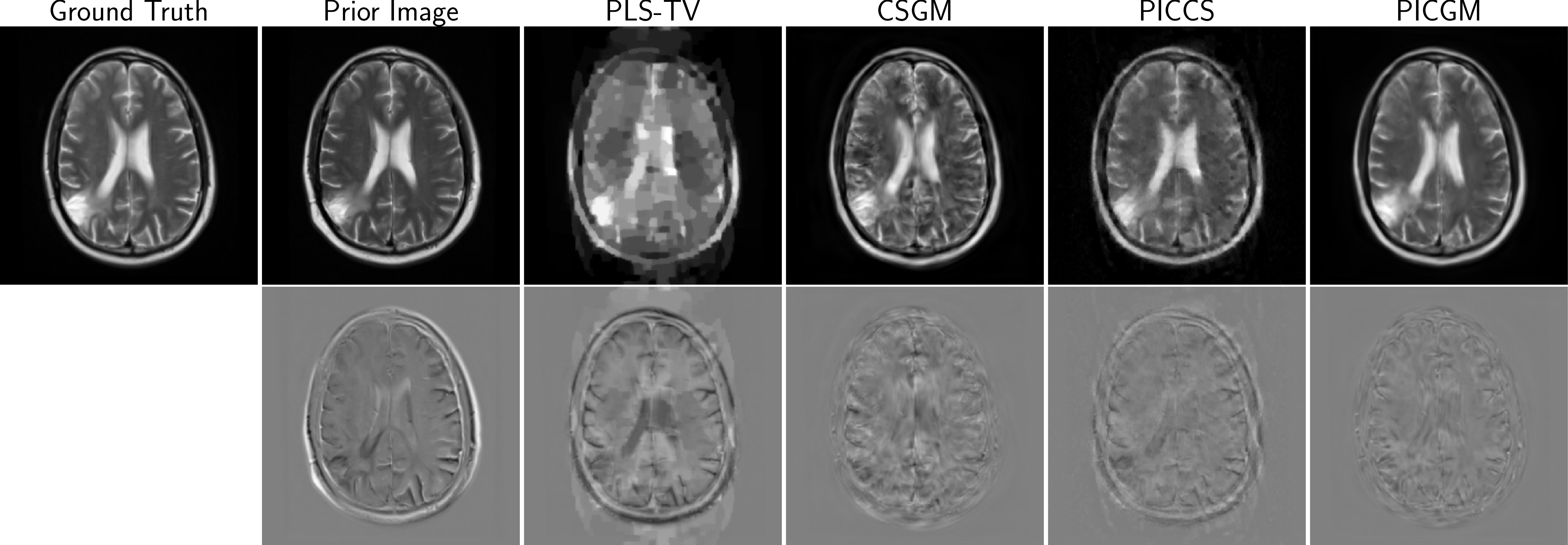}}
\vspace{10pt}
\centerline{\includegraphics[width=\linewidth]{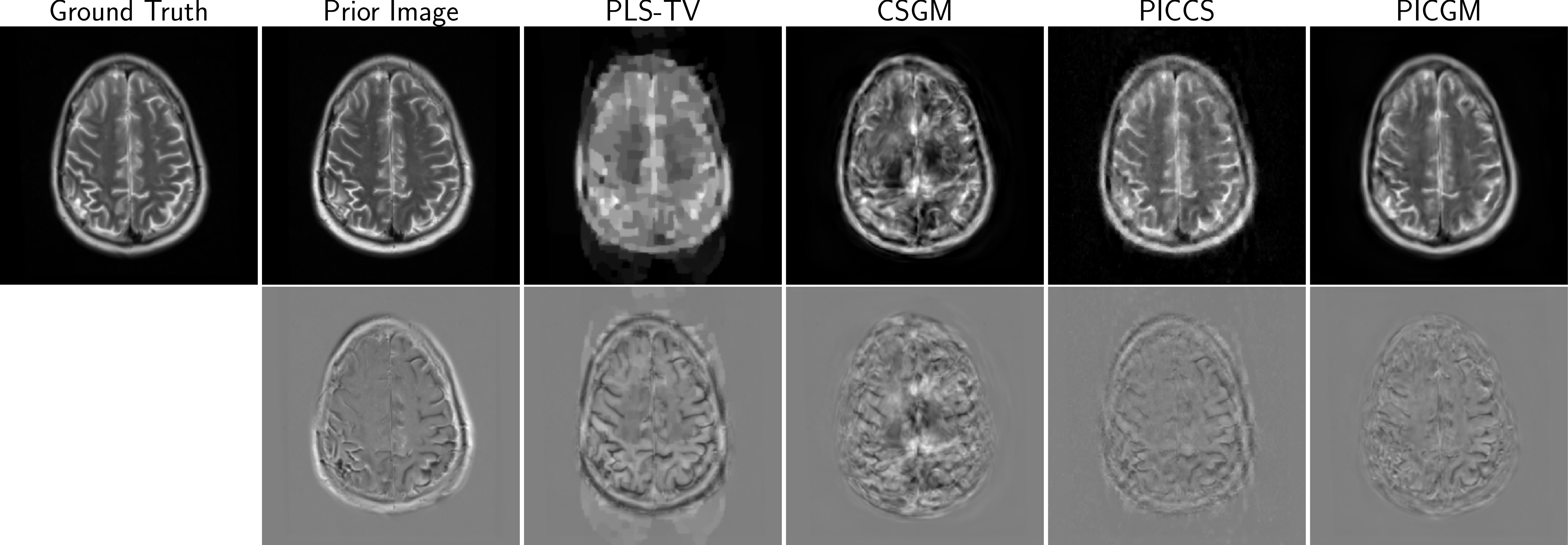}}
\vspace{10pt}
\centerline{\includegraphics[width=\linewidth]{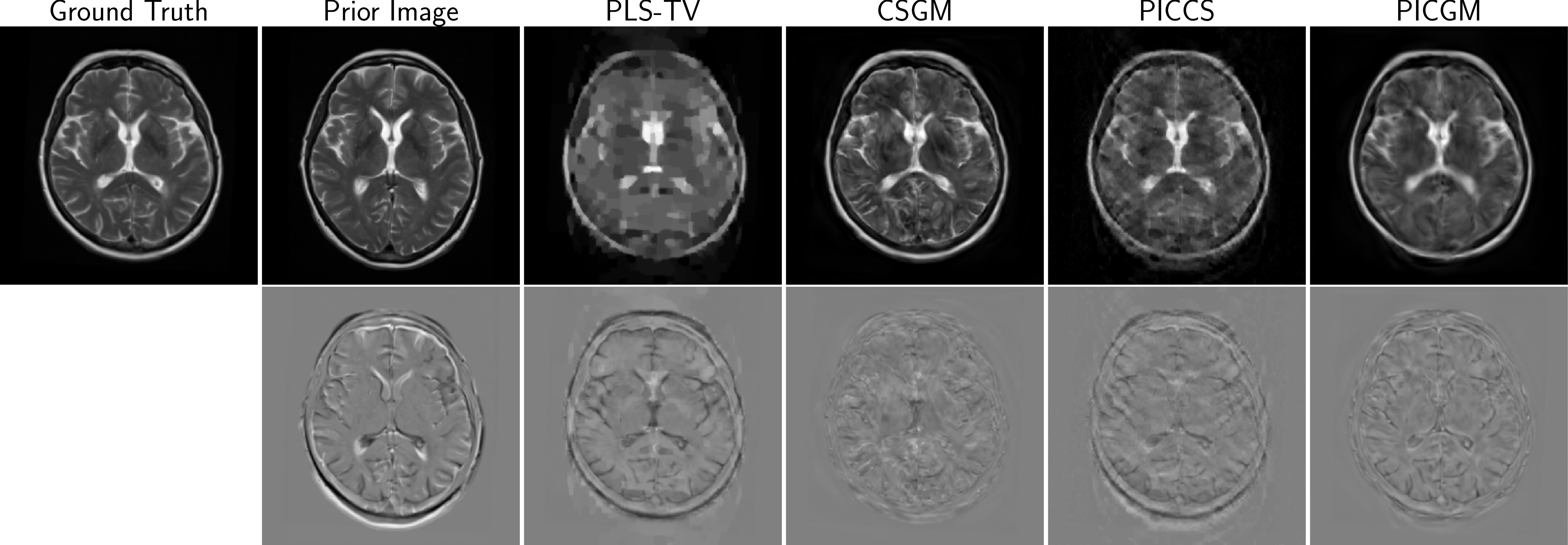}}
\captionof{figure}{Ground truth, prior image, and images reconstructed from simulated MRI measurements with $n/m = 6$ along with difference images for the MR image study}
\label{fig:recon_mask_rand_6x}
\end{center}
\vskip -0.2in

\newpage
\subsection{MR image study: $n/m = 8$}
\vskip 0.2in
\begin{center}
\centerline{\includegraphics[width=\linewidth]{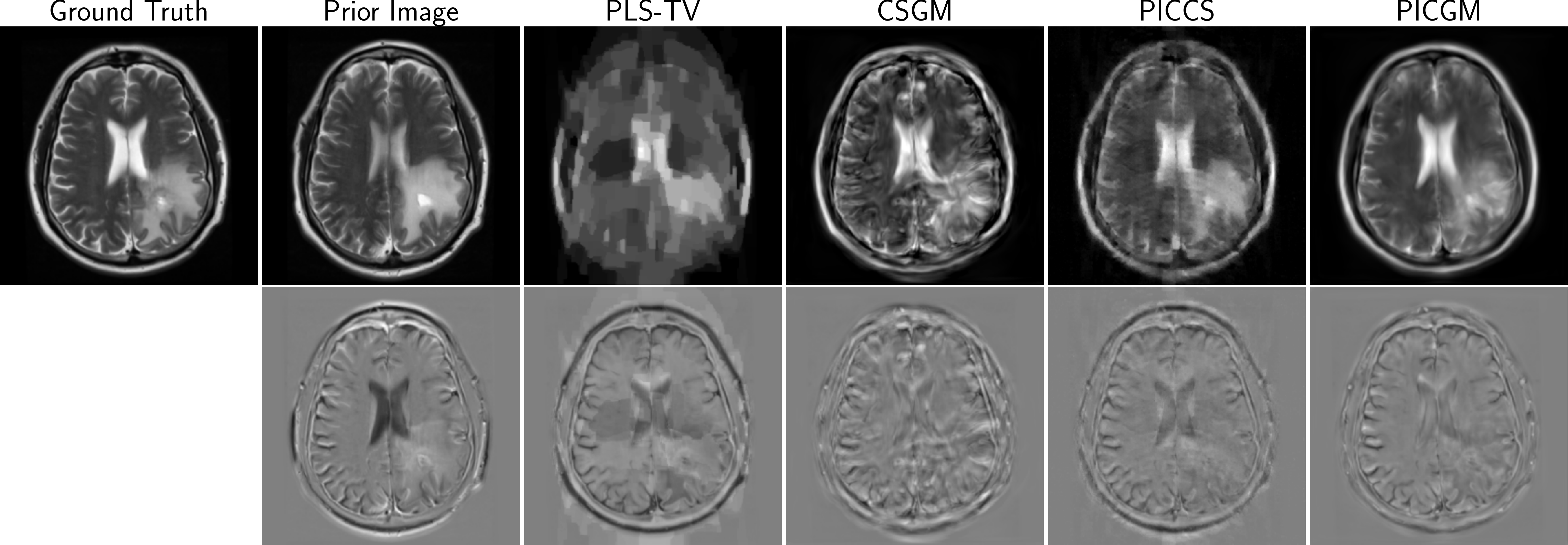}}
\vspace{10pt}
\centerline{\includegraphics[width=\linewidth]{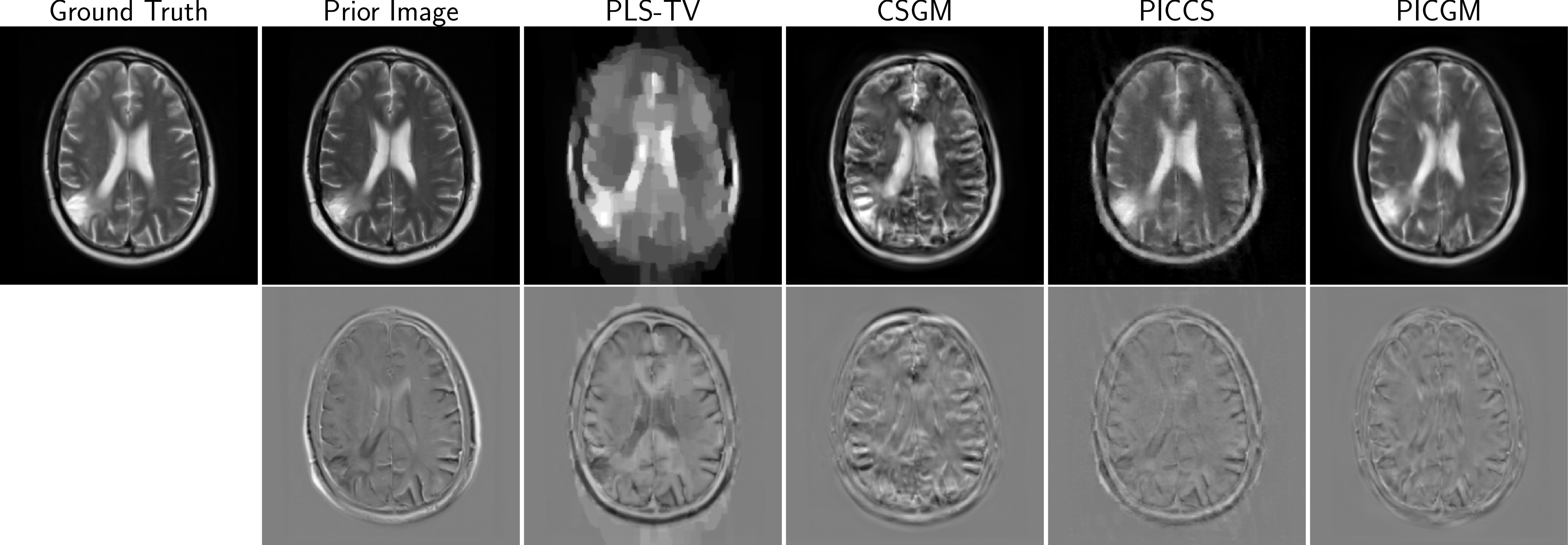}}
\vspace{10pt}
\centerline{\includegraphics[width=\linewidth]{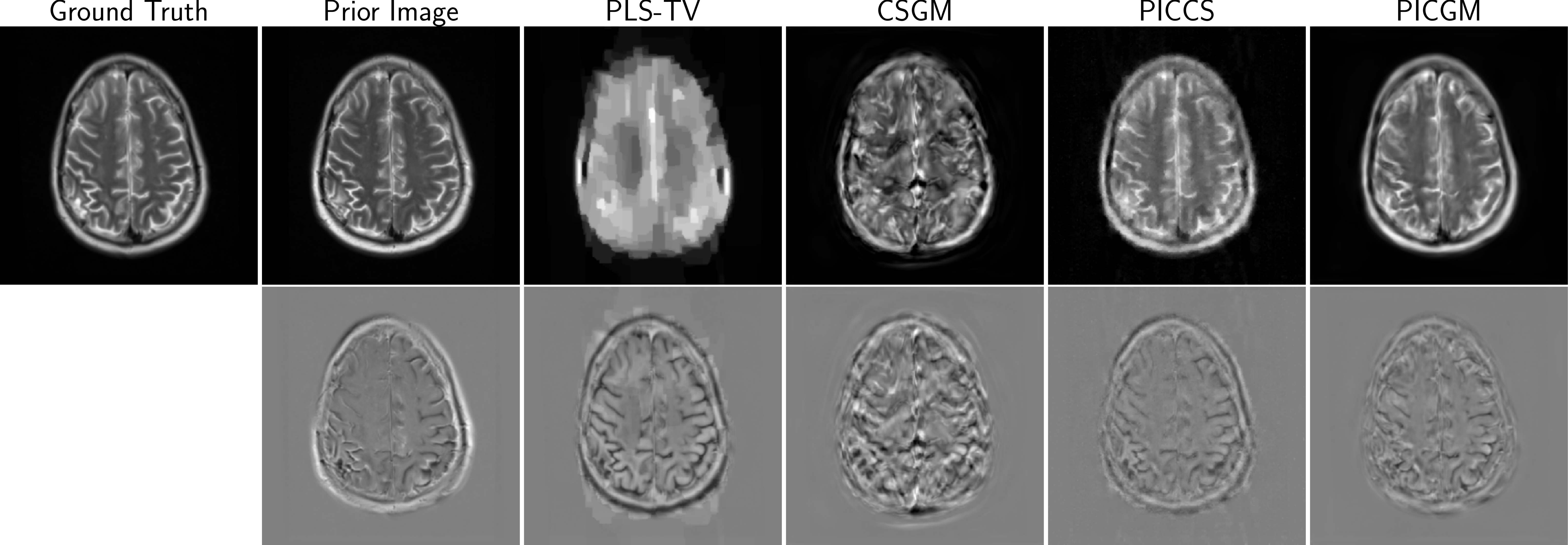}}
\captionof{figure}{Ground truth, prior image, and images reconstructed from simulated MRI measurements with $n/m = 8$ along with difference images for the MR image study}
\label{fig:recon_mask_rand_8x}
\end{center}
\vskip -0.2in

\newpage
\subsection{MR image study: $n/m = 12$}
\vskip 0.2in
\begin{center}
\centerline{\includegraphics[width=\linewidth]{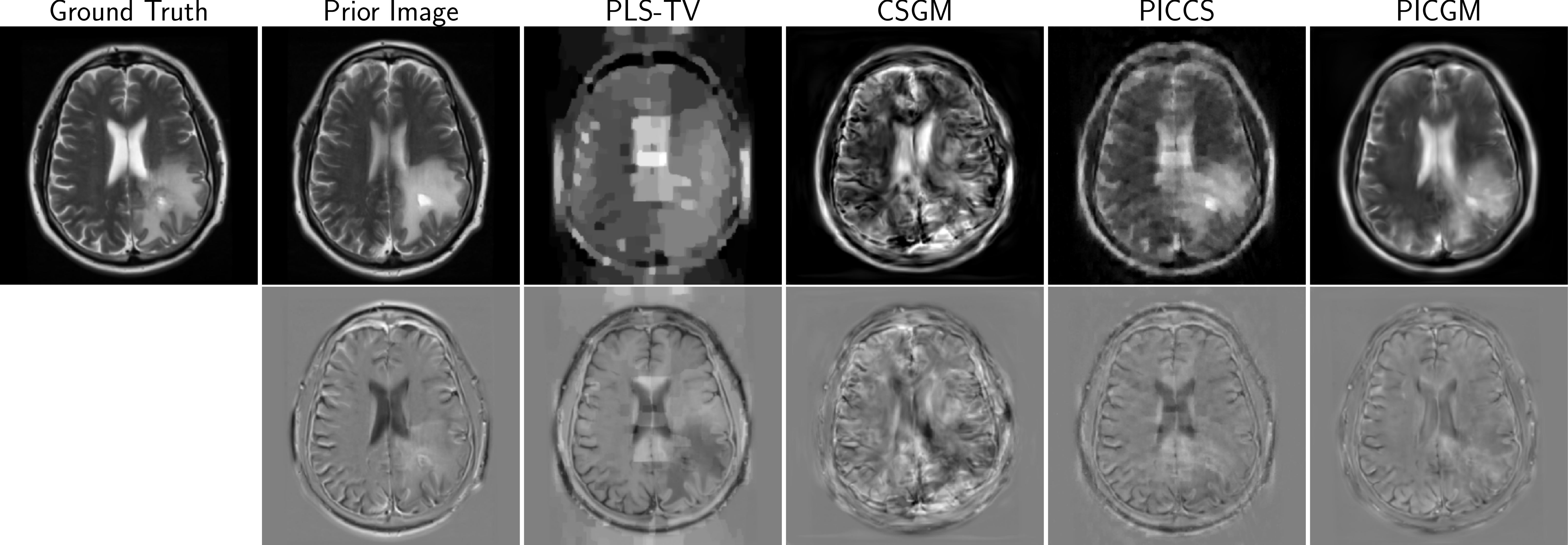}}
\vspace{10pt}
\centerline{\includegraphics[width=\linewidth]{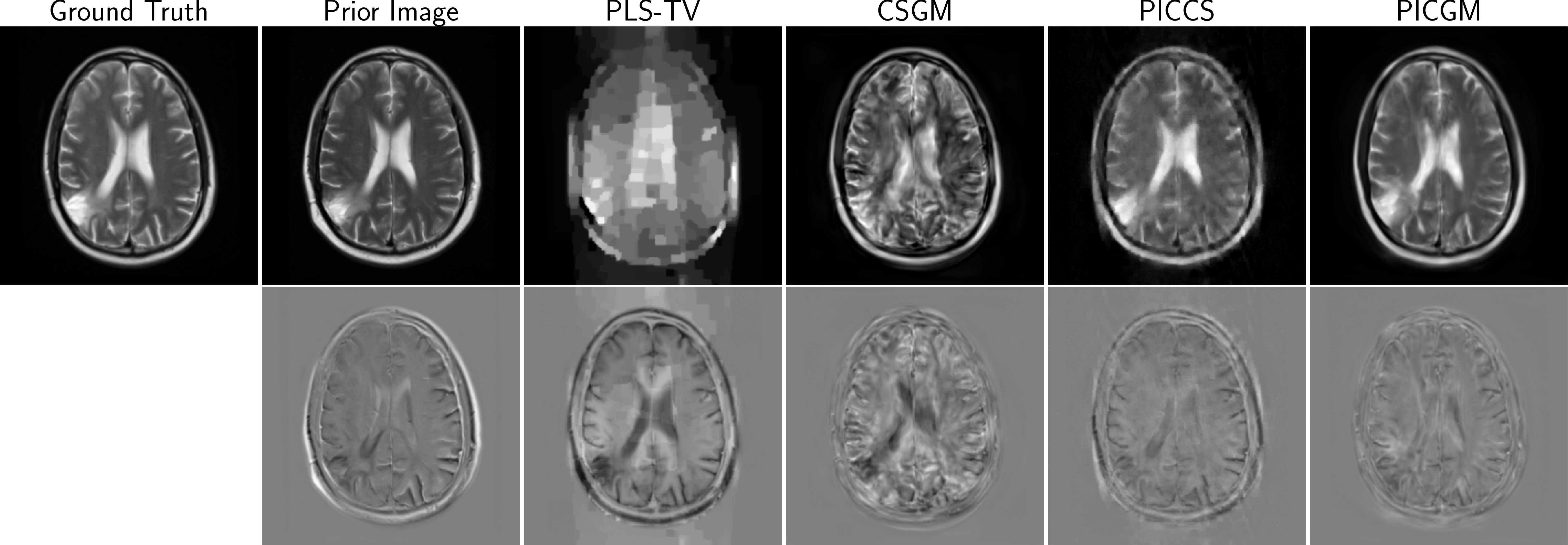}}
\vspace{10pt}
\centerline{\includegraphics[width=\linewidth]{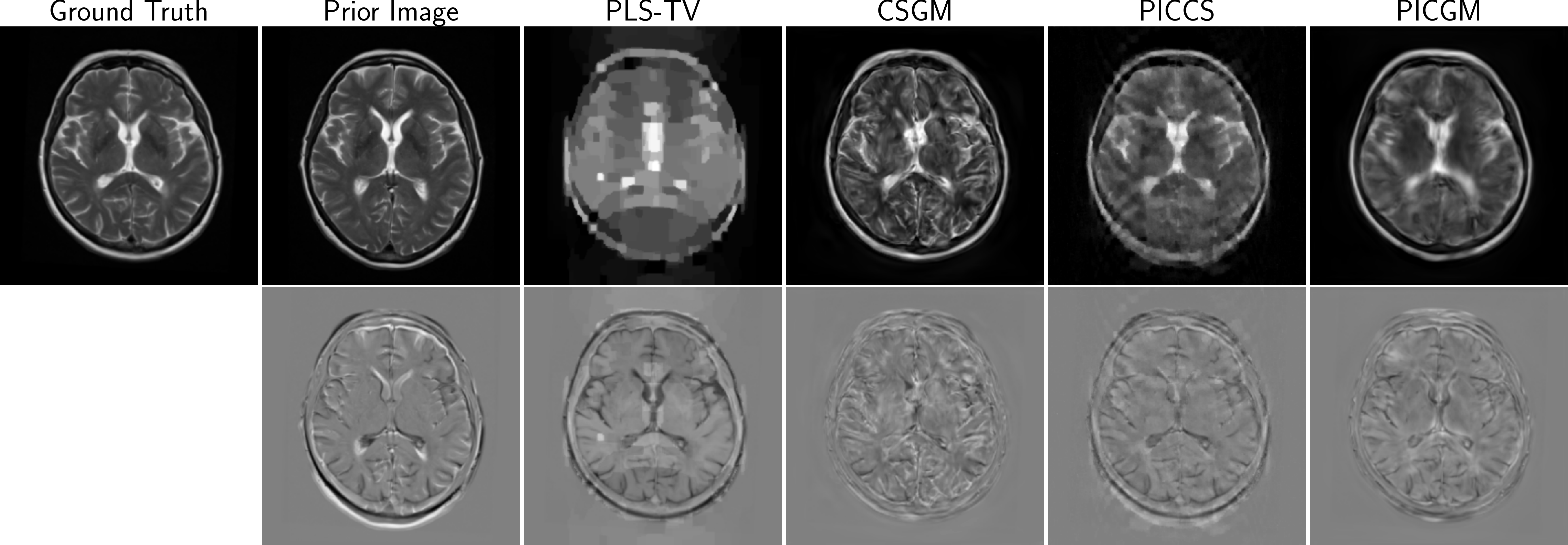}}
\captionof{figure}{Ground truth, prior image, and images reconstructed from simulated MRI measurements with $n/m = 12$ along with difference images for the MR image study}
\label{fig:recon_mask_rand_12x}
\end{center}
\vskip -0.2in

%
%%%%%%%%%%%%%%%%%%%%%%%%%%%%%%%%%%%%%%%%%%%%%%%%%%%%%%%%%%%%%%%%%%%%%%%%%%%%%%%
%%%%%%%%%%%%%%%%%%%%%%%%%%%%%%%%%%%%%%%%%%%%%%%%%%%%%%%%%%%%%%%%%%%%%%%%%%%%%%%

\end{document}